\documentclass[]{interact}

\usepackage{setspace}
\usepackage{amsmath}
\usepackage{graphicx}
\usepackage[toc,page]{appendix}

\usepackage{enumerate}
\usepackage{lmodern}
\usepackage{url}            % simple URL typesetting
\usepackage{booktabs}       % professional-quality tables
\usepackage{amsfonts}       % blackboard math symbols
\usepackage{nicefrac}       % compact symbols for 1/2, etc.
\usepackage{microtype}      % microtypography
\usepackage{lipsum}       % header
\usepackage{graphicx}       % graphics
\graphicspath{{media/}}     % organize your images and other figures under media/ folder
\usepackage{mathtools}
\usepackage{amssymb}
\usepackage{amsthm}
\usepackage{listings} 
\usepackage[caption=false, font=footnotesize]{subfig}
\usepackage{mwe}
\usepackage{placeins}
\usepackage{cleveref}

\newtheorem{theorem}{Theorem}[section]

\begin{document}

\articletype{ARTICLE TEMPLATE}


\title{A new platooning model for connected and autonomous vehicles to improve string stability}

\author{
\name{Shouwei Hui \textsuperscript{a}\thanks{Email: huihui@ucdavis.edu} and Michael Zhang\textsuperscript{b}\thanks{Corresponding author; Email: hmzhang@ucdavis.edu}}
\affil{\textsuperscript{a}Department of Mathematics, University of California Davis, Davis, CA, 95616, United States; \textsuperscript{b}Department of Civil and Environmental Engineering, University of California Davis, Davis, CA, 95616, United States}
}       

\maketitle
 %% Article title


%% Abstract
\begin{abstract}
This paper presents a novel approach to coordinated vehicle platooning, where the platoon followers only communicate with the platoon leader. A dynamic model is proposed to account for driving safety under communication delays. General linear stability results are mathematically proven, and numerical simulations are performed to analyze the impact of model parameters in two scenarios: a ring road with initial disturbance and an infinite road with periodic disturbance. The simulation outcomes align with the theoretical analysis, demonstrating that the proposed "look-to-the-leader" platooning strategy significantly outperforms conventional car-following strategies, such as following one or two vehicles ahead, in terms of traffic flow stabilization. This paper introduces a new perspective on organizing platoons for connected and autonomous vehicles, with implications for enhancing traffic stability.
\end{abstract}


%% Keywords
\begin{keywords}
Platoon; Connected and Autonomous Vehicles; Car-following model; Linear stability
\end{keywords}

\section{Introduction}
% Why platoon model for CAVs
In recent years, Connected and Autonomous Vehicles (CAVs) have acquired significant global attention. With the fast development of CAV technologies and their supporting infrastructures, cooperative car-following (CF), including platooning, are poised to be implemented in the near future to enhance traffic flow. Platooning is a coordinated driving method for a group of vehicles that can be analysed as a systematic longitudinal traffic control system \cite{sheikholeslam1990longitudinal}. Such systems are typically modeled using microscopic traffic models, particularly those based on ordinary differential equations (ODEs). 
% Car following models: basis models

Pipes \cite{pipes1953operational} was among the first to introduce car-following models in 1953, which described the behavior of a string of cars. Since then, numerous car-following models with adjustable parameters have been developed. Bando et al. \cite{bando1995dynamical, bando1995phenomenological} proposed the Optimal Velocity Model (OVM) that replaces the velocity of the lead vehicle in Pipes' model with an optimal velocity function. In particular, OVM can capture traffic instabilities on a ring road without external disturbances. Treiber et al. \cite{treiber2000congested} introduced the Intelligent Driver Model (IDM), which further accounts for the velocity difference of the lead vehicle. These models are widely used in traffic simulations and control design since they are able to represent typical traffic phenomena in relatively simple forms.
% Extension and Control design on CF models and their development on AVs, Sun*2, Zhu, Nakayama, Lenz, Jiang, Zhang, Yu, Treiber 

The aforementioned models only describe the interaction between two vehicles in a leader-follower configuration. Several extensions have been developed based on these models. Lenz et al. \cite{lenz1999multi} introduced a multi-following model based on OVM, which considers the influence of multiple vehicles ahead. Nakayama et al. \cite{nakayama2001effect} proposed a backward-looking model, also based on OVM, to capture the effect of the vehicle behind on the subject vehicle. Jiang et al. \cite{jiang2001full} introduced the Full Velocity Difference Model (FVDM), which accounts for both positive and negative relative velocity to eliminate unrealistic acceleration and deceleration behaviors. Yu et al. \cite{yu2013full} extended FVDM by incorporating acceleration differences. Lazar et al. \cite{lazar2016review} provided a comprehensive review of OVM-based models from 1995 to 2016. Treiber et al. \cite{treiber2006delays} expanded the IDM to follow multiple vehicles ahead while considering reaction time and estimation errors. Derbel et al. \cite{derbel2013modified} proposed a modified version of IDM that enhances vehicle performance and safety.
% Car following models with AVs and platooning, Ngoduy, Zhu, Zhang, Wang, Zhao, Zhou*2

Car-following models can be integrated with various control designs to evaluate the impact of CAVs within a mixed traffic environment consisting of both human-driven and autonomous vehicles. Zhu and Zhang \cite{zhu2018analysis} modified OVM by introducing a smoothing factor to model autonomous vehicles (AVs) and analyzed mixed traffic involving both AVs and human-driven vehicles (HDVs). Jia and Ngoduy \cite{jia2016platoon} developed a platoon-controlled car-following model based on realistic inter-vehicle communication and designed a consensus-based control system for multiple platoons moving in unison. Zhang et al. \cite{zhang2021internet} extended Zhu's model for platooned CAVs by combining multi-following with weighted relative velocity differences and analyzed the linear stability of the model. Wang et al. \cite{wang2021leading} proposed Leading Cruise Control (LCC), a mixed traffic flow control strategy for CAVs that allows them to adaptively follow the vehicle ahead while pacing the following vehicles. Zhao et al. \cite{zhao2023safety} introduced safety-critical traffic control (STC), where CAVs maintain safety relative to both the preceding vehicle and following HDVs, with HDVs modeled using OVM in both Wang's and Zhao's papers. Zhou et al. \cite{zhou2023autonomous, zhou2024self} proposed a platoon-based Intelligent Driver Model (P-IDM) along with an Autonomous Platoon Formation Strategy (APFS) to stabilize CAV platoons under periodic disturbances. Jin and Meng \cite{jin2020dynamical} presented a model combining OVM with a time-delayed feedback controller, increasing stability by analyzing the transcendental characteristic equation to eliminate unstable eigenvalues. Sun et al. \cite{sun2020relationship} examined the relationship between string instability in IDM-based controllers and traffic oscillations, identifying optimal parameters for CAVs to improve stability by forming finite-sized platoons.
% Idea of platooning with other control on platooning (optimization, mpc, RL, etc)  optimization*1, mpc*1, field experiments*3

% Theoretically, other than car-following models, methodologies including stochastic modelling, constrained optimization and model predictive control (MPC) are also applied by traffic researcheres on CAV platooning problems. Li \cite{li2017stochastic} proposed stochastic dynamic model for vehicle platoons and investigated its statistical characteristics. Gong et al. \cite{gong2016constrained} proposed a multi-agent dynamic system to model platoons and the control scheme is formulated as a constrained optimization problem. Zhou et al. \cite{zhou2019distributed} proposed a series of distributed MPC strategies for coordinated car-following of CAVs and gave strategies for guaranteed local, $L_2$, and $L_{\infty}$ stabilities. 

Theoretically, beyond car-following models, methodologies such as stochastic modeling, constrained optimization, and model predictive control (MPC) have also been applied by traffic researchers to address CAV platooning problems. Li \cite{li2017stochastic, li2017stochastic2} proposed a stochastic dynamic model for vehicle platoons and investigated its statistical characteristics. Gong et al. \cite{gong2016constrained} introduced a multi-agent dynamic system to model platoons, with the control scheme formulated as a constrained optimization problem. Zhou et al. \cite{zhou2019distributed} developed a series of distributed MPC strategies for coordinated car-following of CAVs, providing approaches that ensure local, $L_2$, and $L_{\infty}$ stabilities. Graffione et al. \cite{graffione2020model} introduced an MPC approach to control inter-vehicular distances and speed within a vehicle platoon, which improves safety and reduces fuel consumption. Li et al. \cite{li2022variable} developed a platoon control strategy for heterogeneous connected vehicles (CVs) by incorporating a variable time headway (VTH) spacing policy and a sliding mode controller. Liu et al. \cite{liu2022autonomous} adopted an integrated deep reinforcement learning (DRL) and dynamic programming (DP) approach to learn autonomous platoon control policies.

Field experiments and related data analysis can further demonstrate the effectiveness of CAVs and CAV platoons in stabilizing traffic and reducing fuel consumption. Stern et al. \cite{stern2018dissipation} conducted a ring road experiment with over 20 vehicles, showing that a controlled autonomous vehicle can dampen stop-and-go waves, reducing total fuel consumption. Lee et al. \cite{lee2024traffic} implemented a large-scale field experiment involving 100 CAVs functioning as a MegaController on a freeway network to mitigate stop-and-go waves, marking the largest CAV field experiment to date. Tsugawa et al. \cite{tsugawa2011automated} tested a platoon of three trucks on a test track along an expressway, demonstrating a $14\%$ reduction in fuel consumption. By combining experimental datasets, Zhou et al. \cite{zhou2023experimental} investigated the emissions and fuel consumption (EFC) characteristics in car-following (CF) platoons.
% Limitation, Contribution and structure

In general, previous studies on CAV platoon control have demonstrated the potential of CAVs in stabilizing traffic, particularly by developing car-following models that enhance traffic stability through advanced control and communication strategies. However, most existing models rely on communication with multiple vehicles, which increases system complexity and communication load, especially in mixed-traffic conditions. Additionally, few studies prioritize a platoon structure that optimally reduces inter-vehicle communication requirements while maintaining robust stability. Motivated by these research gaps, in this paper we make several contributions to CAV platoon control. We propose a novel platoon CF model that enhances traffic stability by focusing on coordinated communication strategies, particularly prioritizing the role of the leading vehicle. We also show the superiority of our models by performing rigorous mathematical stability proofs and comprehensive numerical simulations. The simulations cover a variety of traffic scenarios, demonstrating how different parameters can impact the stability of a single platoon. Our findings indicate that the proposed control strategy significantly outperforms conventional methods in suppressing traffic oscillations and maintaining string stability.

The remaining parts of this paper are organized as follows. In section \ref{sec2}, the car-following models for HDVs and platoon-controlled CAVs are introduced. In section \ref{sec3}, linear stability criteria of the proposed models are presented and proved. In section\ref{sec4}, numerical simulations on a ring road and an infinite road are demonstrated and performances of different model parameters are evaluated. Lastly in section \ref{sec5}, the conclusion is drawn and potential directions for future research are discussed.

\section{The Non-Platoon and Platoon CF Models \label{sec2}}
Without loss of generality, we assume a string of $N$ vehicles are moving along a single-lane road without overtaking. The initial position of the rear bumper of the first vehicle is at $0$m on the road, with the leading vehicle designated as $N-$th vehicle. In particular, if the single-lane road is a ring road, the $(N+1)$-th vehicle is the same as the $1$st vehicle. To demonstrate the effectiveness of the platoon-controller for connected and autonomous vehicles (CAV), we introduce a base model for human-driven vehicles (HDVs) in Section \ref{base}. In Section \ref{ovmsec2}, we present the model for platoon-controlled CAVs, followed by the transition phase model in Section \ref{ovmsec3}.

\subsection{The optimal velocity model \label{base}}
A commonly used car-following model to represent human-driven vehicles (HDVs) is the Optimal Velocity Model (OVM) \cite{bando1995dynamical}. OVM has the form
\begin{equation}
    \ddot{x}_i(t)=a\left[V(x_{i+1}(t)-x_{i}(t))-\dot{x}_i(t)\right], \label{OVM1}
\end{equation}
where $x_i(t)$, $i=1,2,\dots N$ is the position of the $i$-th vehicle at time $t$ on the single-lane road without overtaking. $x_{i+1}(t)-x_{i}(t)\triangleq h_i(t)$ is the headway between the $i$-th and $i+1$-th vehicle. $\dot{x}_i(t)$, $\ddot{x}_i(t)$ are velocity and acceleration of $i$-th vehicle at time $t$, $V(h)$ is the optimal velocity function of headway (head to head distance) $h$, and $a$ is a sensitivity constant. An example of optimal velocity function is given in \cite{wang2021leading}, which is equivalent to the form
\begin{equation}
    V(h)=\begin{cases}
        v_{\max}, & \text{if } h\geq h_{\max}; \\
        \frac{v_{\max}}{2}\left(1-\cos\left(\pi\frac{h-h_{\min}}{h_{\max}-h_{\min}}\right)\right), & \text{if } h_{\min}\leq h \leq h_{\max};\\
        0, & \text{if } h \leq h_{\min},
    \end{cases}
    \label{ovf1}
\end{equation}
where $h_{\min}$ is the minimum headway, $h_{\max}$ is the maximum headway, $v_{\max}$ is the maximum velocity and $l$ is the length of each vehicle. Figure \ref{f1} is an example plot of \ref{ovf1} and the corresponding fundamental diagram (density-flow diagram).
\begin{figure}
    \centering
\subfloat[Optimal velocity function]{
    \includegraphics[width=0.4\textwidth]{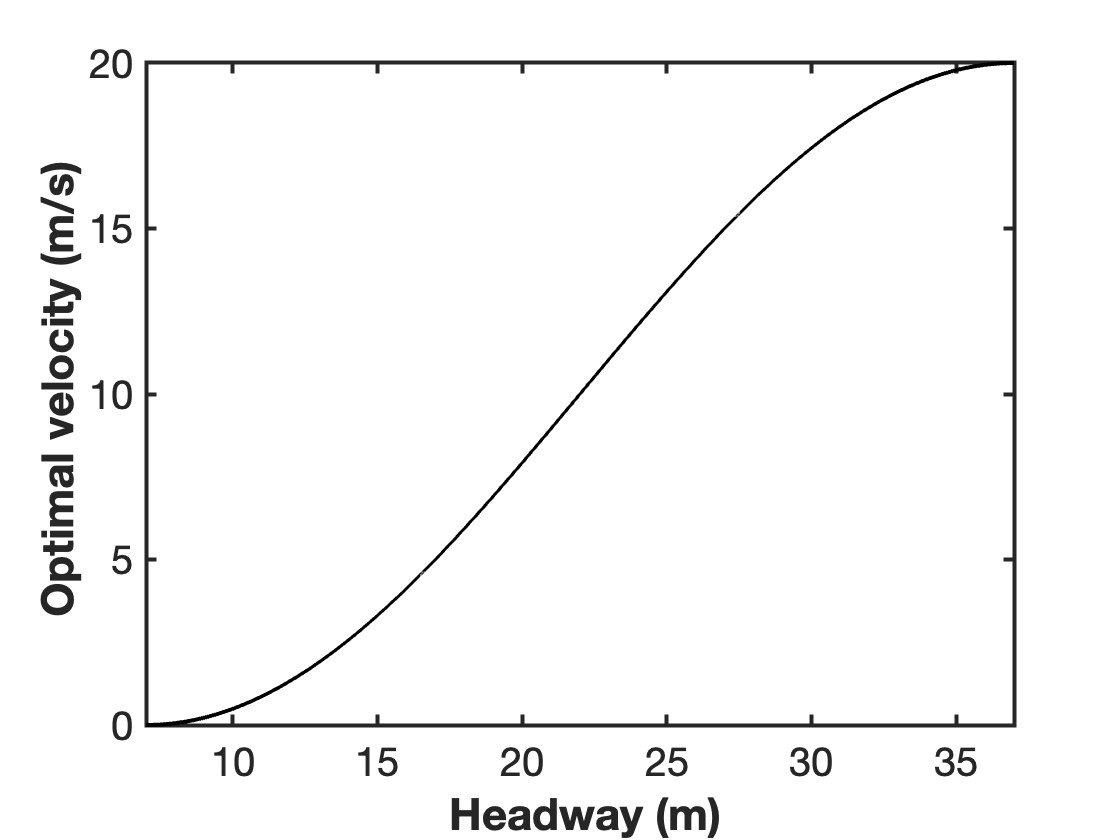}
}
\hfill
\subfloat[Fundamental diagram]{
    \includegraphics[width=0.4\textwidth]{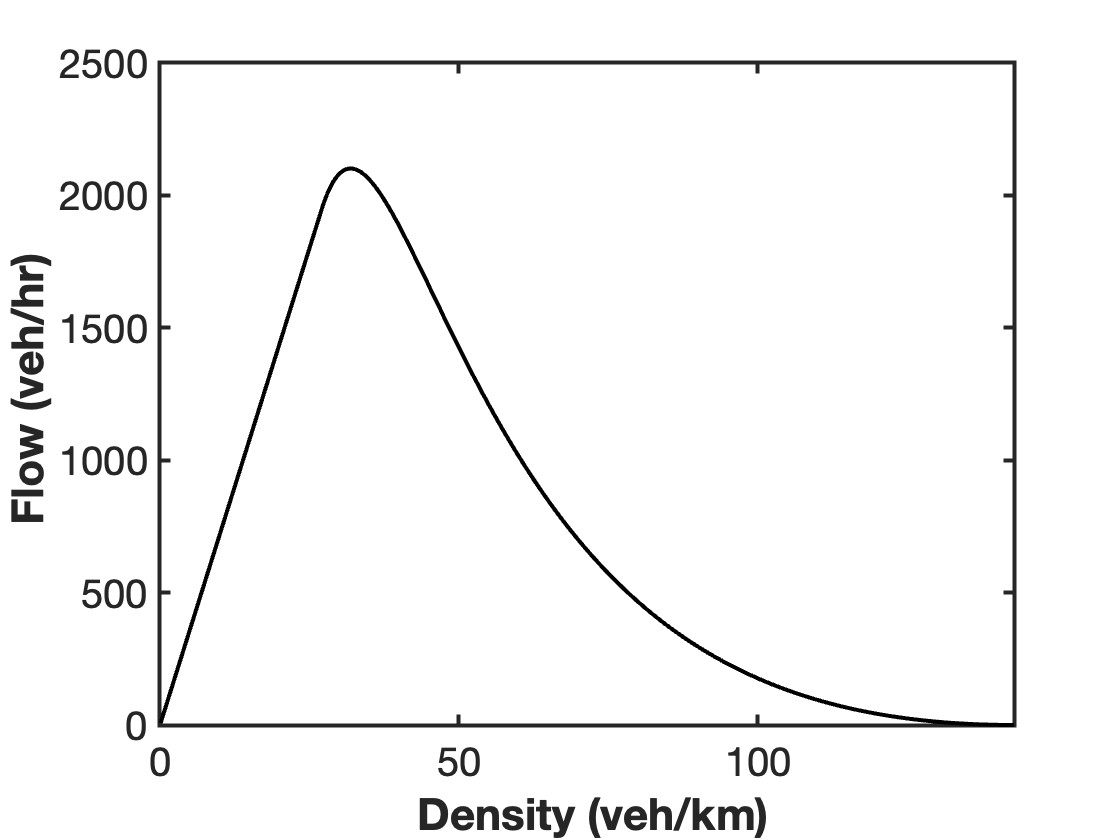}
}
    \caption{Plot of an optimal velocity function with $l=5$m and $v_{\max}=20$m/s and the corresponding fundamental diagram.}
    \label{f1}
\end{figure}

\subsection{Platoon controller: A coordinated car-following strategy \label{ovmsec2}} 
We now consider a system of CAVs where every following vehicle is connected to the leading autonomous vehicle. In this setup, the following vehicles have access to the position and speed of the leading vehicle. In the proposed platoon model, only the position information of the leading vehicle is required. We further assume that the target (optimal) velocity of a following vehicle in the platoon is based on its distance from the leading vehicle:
\begin{equation}
    \ddot{x}_i(t)=a\left[V(\frac{x_{N}(t)-x_{i}(t)}{N-i})-\dot{x}_i(t)\right],\label{povm}
\end{equation}
where $x_N$ is the position of the controlled leading vehicle. We refer \eqref{povm} as the platoon controlled OVM (P-OVM) model. In this model, the spacing between the platoon leader (the $N$-th vehicle) and the $i$-th following vehicle is scaled and used in the optimal velocity function. This is equivalent to considering the average spacing of vehicles between the platoon leader and the $i$-th following vehicle. Figure \ref{d1} illustrates this platoon-controlled car-following configuration. Notably, under this setup, the influence of the immediate preceding vehicle is neglected, which could potentially result in collisions if position measurement errors or communication delays between the platoon leader and the following vehicles reach a critical threshold.
\begin{figure}
    \centering
    \includegraphics[width=0.95\textwidth]{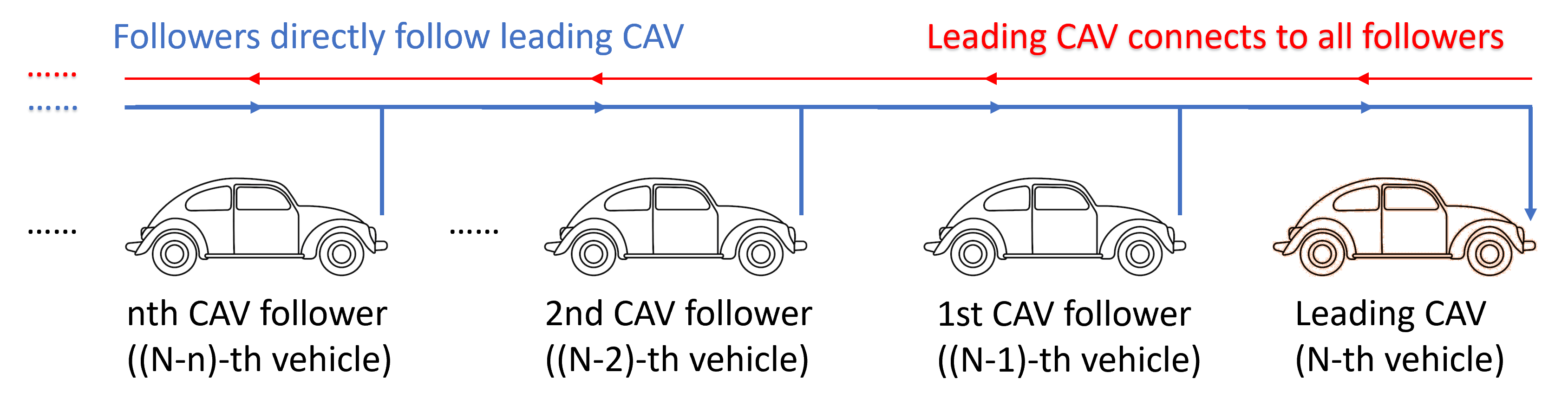}
    \caption{Illustration of the platoon-controlled car-following system. Leading CAV is connected to all followers of the platoon and followers are directly following the leading CAV}
    \label{d1}
\end{figure}

\subsection{The transition phase model \label{ovmsec3}}
We can address the previously mentioned safety issue by adding a local, within platoon vehicle-to-vehicle control to P-OVM. This local control can be any feedback mechanism that prevents vehicle collisions. For the sake of simplicity in presentation and analysis, we adopt OVM as the local control mechanism, referring to the combined model as the Transition Phase Optimal Velocity Model (T-OVM). The model is of the form
\begin{equation}
    \ddot{x}_i=a(V(x_{i+1}-x_i)-\dot{x}_i)+b\left(V\left(\frac{x_N-x_i}{N-i}\right)-\dot{x}_i\right),
    \label{tran}
\end{equation}
where $a$ is the sensitivity to vehicle in front and $b$ is the sensitivity to the leading vehicle. The total sensitivity is defined as the sum of the sensitivities $(a+b)$. Figure \ref{d2} illustrates the car-following pattern of the T-OVM. With this two-part model, each vehicle within the platoon balances between maintaining the optimal velocity relative to the vehicle directly in front and the optimal velocity relative to the platoon leader. This balance helps prevent collisions while simultaneously enhancing stability.
% Another illustration figure
\begin{figure}
    \centering
    \includegraphics[width=0.95\textwidth]{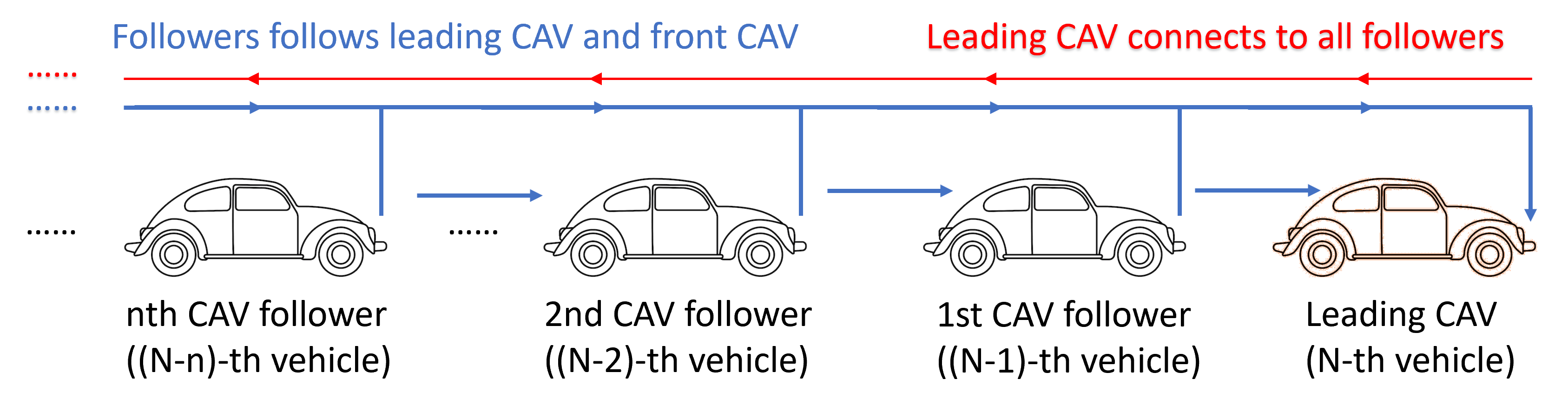}
    \caption{Illustration of the transition phase car-following system. Leading CAV is connected to all followers inside the platoon and followers are following the leading CAV and the front CAV simultaneously}
    \label{d2}
\end{figure}

\section{Stability analysis of the platoon models}\label{sec3}
If the string of $N$ vehicles are on a ring road with length $L$ and the leading vehicle is following the last vehicle of the string by the original OVM \eqref{OVM1}, the position of $i-th$ vehicle is the same for both OVM \eqref{OVM1} and P-OVM \eqref{equi} under steady state:
\begin{equation}
    e_i(t)=hi+V(h)t,
    \label{equi}
\end{equation}
where $h=L/N$ is the equilibrium headway. For the classic optimal velocity model, we have the following linear stability criterion: 
\begin{theorem}
    The optimal velocity model \eqref{OVM1} is linearly stable if 
    \begin{equation}
        a>2V'(h).
        \label{ovmsta}
    \end{equation}
    \label{lin}
\end{theorem}
A proof slightly different from \cite{bando1995dynamical} is given in Appendix \ref{ap1}. Then with the platoon controller applied, it turns out that no stability criteria is required:
\begin{theorem}
    In the $N$-vehicle platoon where the leader vehicle follows the last vehicle of the platoon on a ring road according to OVM \eqref{OVM1}, the platoon-controlled system P-OVM \eqref{povm} is always linearly stable.
    \label{sta2}
\end{theorem}
\begin{proof}
    From previous introduction, P-OVM has the same equilibrium solution of the optimal velocity model as $x_n^0(t)=e_n(t)$ for any car $n$ inside the platoon. Now suppose that all the vehicles are deviated from the initial position with a small disturbance $y_n(t)$, that is,
    \begin{equation}
        x_n(t)=e_n(t)+y_n(t),\;|y_n|\ll 1.
    \end{equation}
    
    To linearize the original system, we can do a Taylor expansion of the optimal velocity function term $V(\Delta x_n)$ and neglect the higher order terms to get
    \begin{equation}
        \ddot{y}_n(t)=\begin{cases}
            a\left[V'(h)\frac{y_{N}(t)-y_n(t)}{N-n}-\dot{y}_n(t)\right],\;\text{if } n\neq N; \\
            a\left[V'(h)(y_{1}(t)-y_N(t))-\dot{y}_N(t)\right],\;\text{if } n= N.
        \end{cases}
        \label{linear}
    \end{equation}
    
    Then by the formation of the controller, $y_1$ and $y_N$ forms a system of 2 linear ODEs, and the solution in vector form is:
    \begin{equation}
        \mathbf{Y}=\sum_{n=1,\;N}C_ne^{\lambda_n t}\mathbf{V}_n,
    \end{equation}
    where $Y = [y_1, y_N]^T$ is the vector form, $\lambda_n$ and $V_n$ are the eigenvalues and eigenvectors of the system, and $C_n$ are constants determined by the initial condition. Now suppose that $\lambda$ is an eigenvalue of the system, and $\xi_n$ is the coefficient of $y_n$ with the term $e^{\lambda t}$, then simplified from \eqref{linear}, $\lambda$ satisfies
    \begin{equation}
        \lambda^2+a\lambda-aV'(h)(\frac{\xi_{1}}{\xi_N}-1)=0,
    \end{equation}
    
    and
    \begin{equation}
        \lambda^2+a\lambda-aV'(h)(\frac{\xi_{N}}{(N-1)\xi_1}-\frac{1}{N-1})=0.
    \end{equation}
    
    Then by comparing these equations we can calculate that $\xi_1=\xi_N$ or $\xi_1=-\xi_N/(N-1)$. And the first case is only true if we have $y_1=y_N$ and the two cars just stay at the same distance from the original equilibrium. For the other case we have
    \begin{equation}
        \lambda^2+a\lambda-aV'(h)(-\frac{1}{N-1}-1)=0.
    \end{equation}
    Then, by solving the quadratic equation, we have
    \begin{equation}
        \lambda= \frac{-a\pm\sqrt{a^2+4aV'(h)(-\frac{1}{N-1}-1)})}{2}.
        \label{cond}
    \end{equation}
   For the system to be stable, we need to have the real parts of both $\lambda$ to be negative. This is true if $V'(h)>0$ which is always true for any well-defined $V(h)$. Now we just need to check if the same property holds for the remaining vehicles. By \eqref{linear} we can have a initial guess that the nonlinear part is $y_j=\zeta y_N$ for $j\neq 1$, then plug this into \eqref{linear} and compare the coefficients we have
    \begin{equation}
        \zeta=N-j+\frac{j-1}{N-1}.
    \end{equation}
    Thus the other cars also have the same eigenvalue for the nonlinear part and the system will always be linearly stable if the problem is well defined.
\end{proof}
As for the transition phase model, stability clearly depends on both $a$, $b$ and $N$. For $N$ sufficiently large we have the following stability criterion:
\begin{theorem}
    Assuming that the leading vehicle follows the last vehicle by OVM \eqref{OVM1} with sensitivity constant $a+b$. Then the transition phase model T-OVM \eqref{tran} is linearly stable if
    \begin{equation}
        \frac{(a+b)^2}{a}>2V'(h).
        \label{transtae}
    \end{equation}
    \label{transta}
\end{theorem}
The proof is given in Appendix \ref{ap2}. Figure \ref{s1} shows the plot of neutral stability lines for different sensitivity values. The area above the lines represents stability, while the area below indicates instability. From the figure, it is evident that as the value of $b$ increases, the unstable area becomes smaller and eventually disappears when $a\to0$.
\begin{figure}
    \centering
    \includegraphics[width=0.5\textwidth]{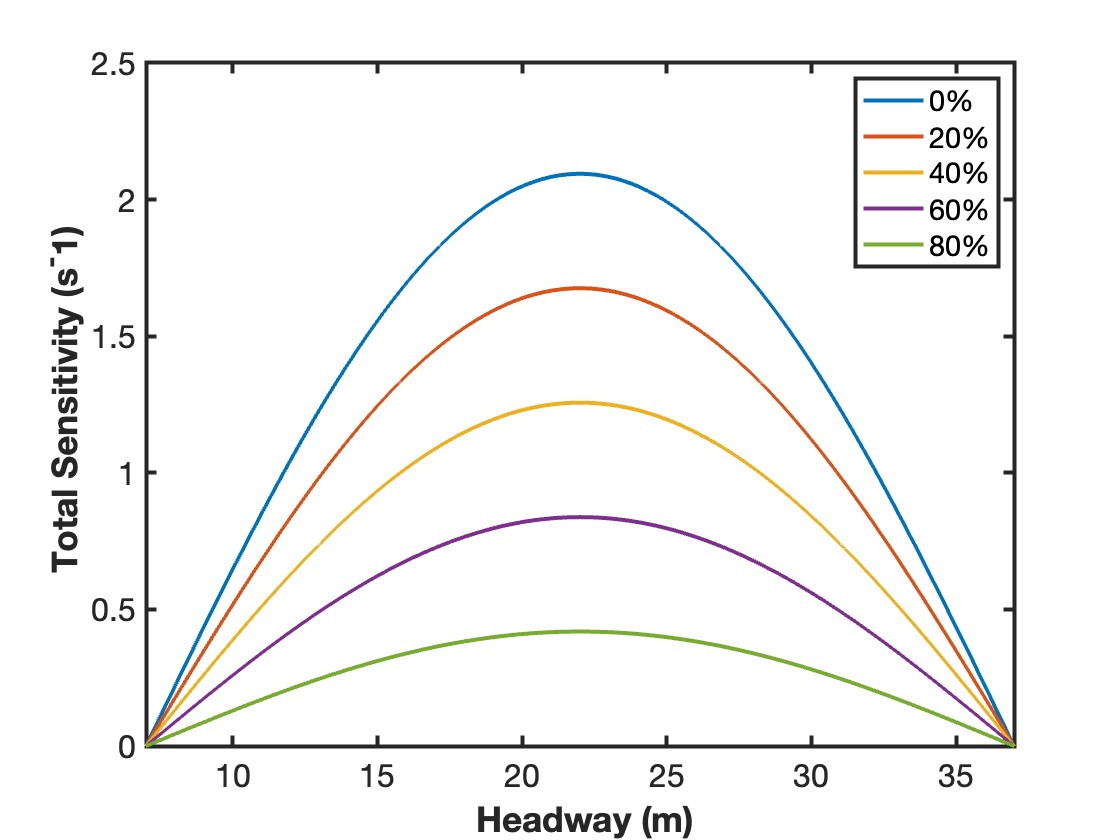}
    \caption{neutral stability lines of the transition phase model for selected percentages of leading vehicle sensitivity $b/(a+b)=0\%, 20\%, 40\%, 60\%, 80\%$. Above each line the region is stable and below is unstable.}
    \label{s1}
\end{figure}
%A plot of stability region%
\section{Numerical simulations \label{sec4}}

To better visualize how stability improves with the application of platoon control, several numerical experiments are conducted using different road configurations and model parameters. All simulations are performed in MATLAB 2023b, with a time step of $\Delta t=0.1$s. The forward Euler scheme is used to update the velocity variable $\dot{x}_i$ and a modified Euler scheme (Heun's method) is employed to update the position variable $x_i$:
\begin{equation}
\begin{cases}
    \dot{x}_{i,j+1}=\dot{x}_{i,j}+\ddot{x}_{i,j}\Delta t;\\
    x_{i,j+1}=x_{i,j}+\frac{\dot{x}_{i,j}+\dot{x}_{i,j+1}}{2}\Delta t,
\end{cases}
\end{equation}
where $x_{i, j}$ is referring to the position of $i$-th car at $j$-th time step of simulation. This is equivalent to the discretizaiton method in \cite{zhu2018analysis}.

\subsection{Ring road with initial disturbance}
In this subsection, simulations are conducted on a ring road with a total length of the ring road is $L=264$m with 12 vehicles. For model parameters, we use the optimal velocity function given by \eqref{ovf1} setting $h_{\min}=7$m, $h_{\max}=37$m, $v_{\max}=20$m/s and $l=5$m. This results in an equilibrium headway $h=L/N=22$m and an equilibrium velocity $V(h)=10$ m/s. The initial position and velocity of the $i$-th vehicle deviate from the equilibrium states $(e_i,V(h))$ with random perturbation uniformly distributed on the interval $[0,5]$. The initial condition of the model can be written as
\begin{equation}
    \begin{cases}
        x_i(0) =& e_i(0)+r_i; \\
        \dot{x}_i(0) = & V(h)+\bar{r}_i,
    \end{cases}
\end{equation}
where $r_i,\bar{r}_i$ are random numbers generated from uniform distribution on $[0,5]$, and $e_i(0)=hi$ can be calculated from \eqref{equi}.

\textbf{Simulation 1.1}: Comparison between OVM and P-OVM

To demonstrate the stability improvement of P-OVM, we tested sensitivity constants $a=0.4,\;0.8,\;1.6,\;2.4$ for both OVM and P-OVM  under the same initial conditions. Simulation results are shown for the original OVM and platoon-controlled OVM in Figures \ref{11}-\ref{14}. Figures \ref{11}, \ref{12} are the 3-D plots of headways of all vehicles under OVM and P-OVM, respectively, with different values of $a$. Figures \ref{13}, \ref{14} are the velocity profiles of the $6$th vehicle under OVM and P-OVM, respectively, with different values of $a$.
%Plots of this simulation% 

\begin{figure}
    \centering
    \subfloat[$a=0.4$]{
        \includegraphics[width=0.4\textwidth]{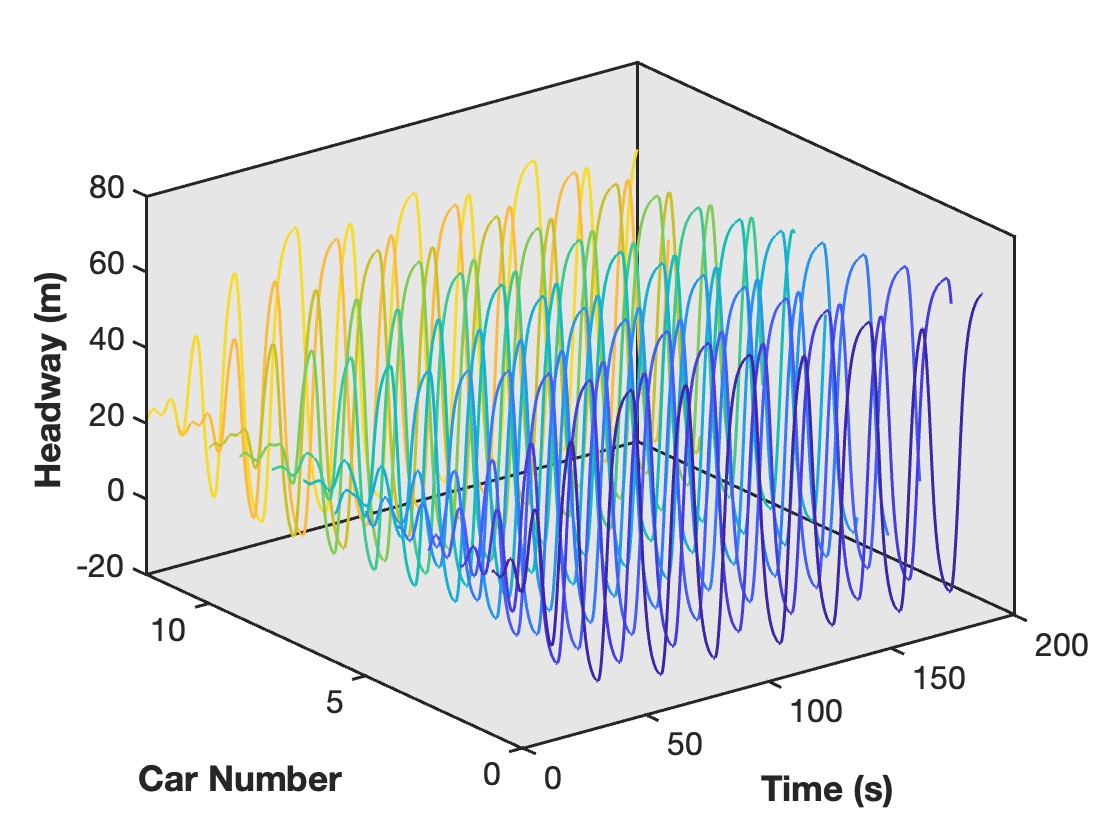}
    }
    \subfloat[$a=0.8$]{
        \includegraphics[width=0.4\textwidth]{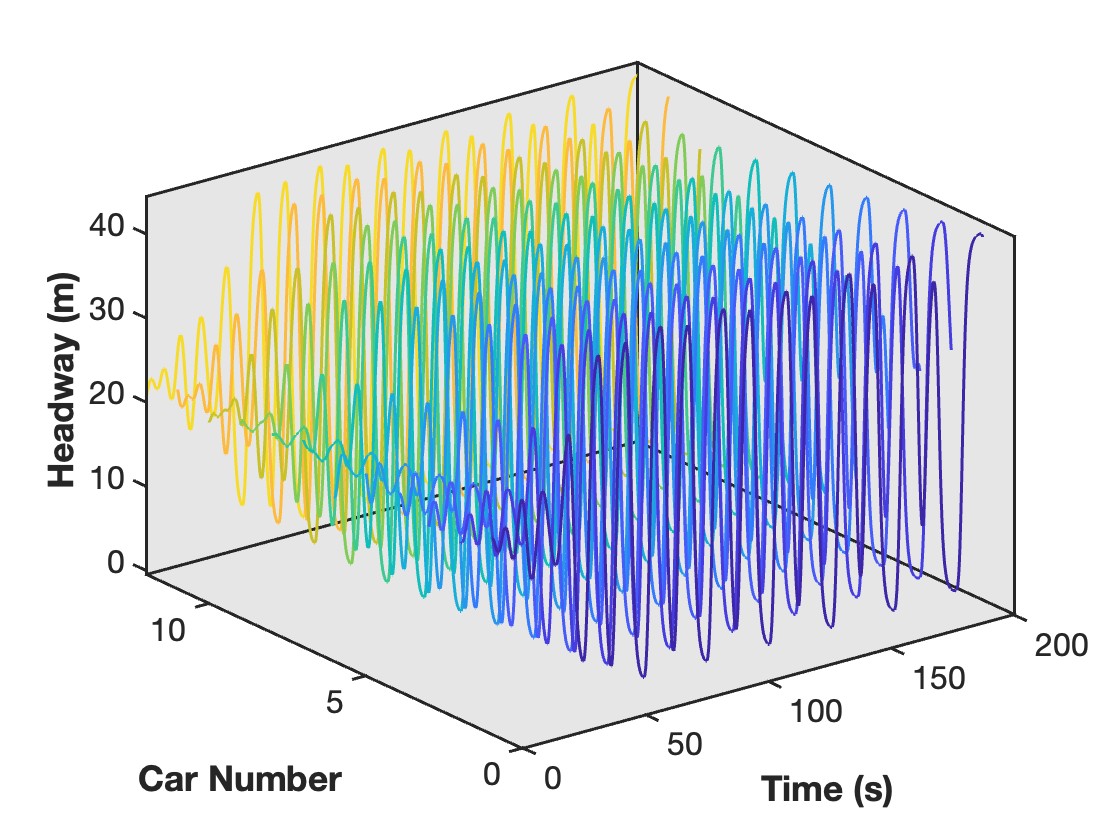}
    }
    \\
    \subfloat[$a=1.6$]{
        \includegraphics[width=0.4\textwidth]{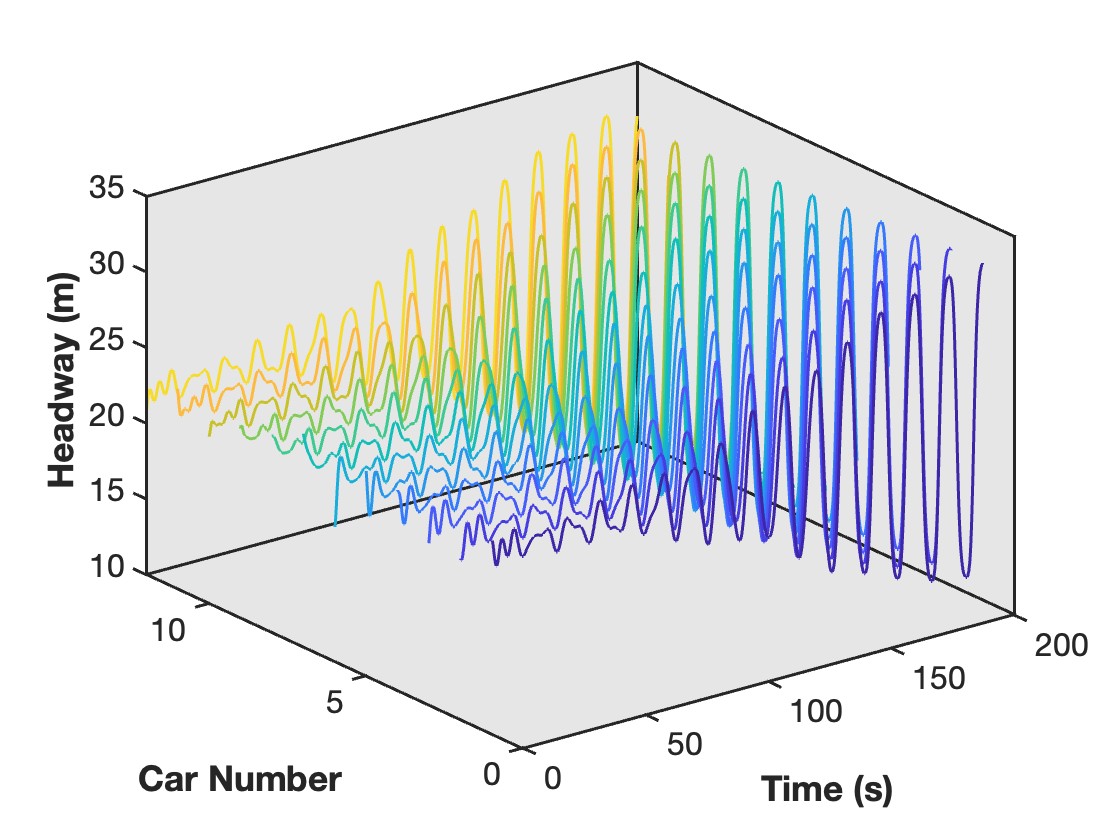}
    }
    \subfloat[$a=2.4$]{
        \includegraphics[width=0.4\textwidth]{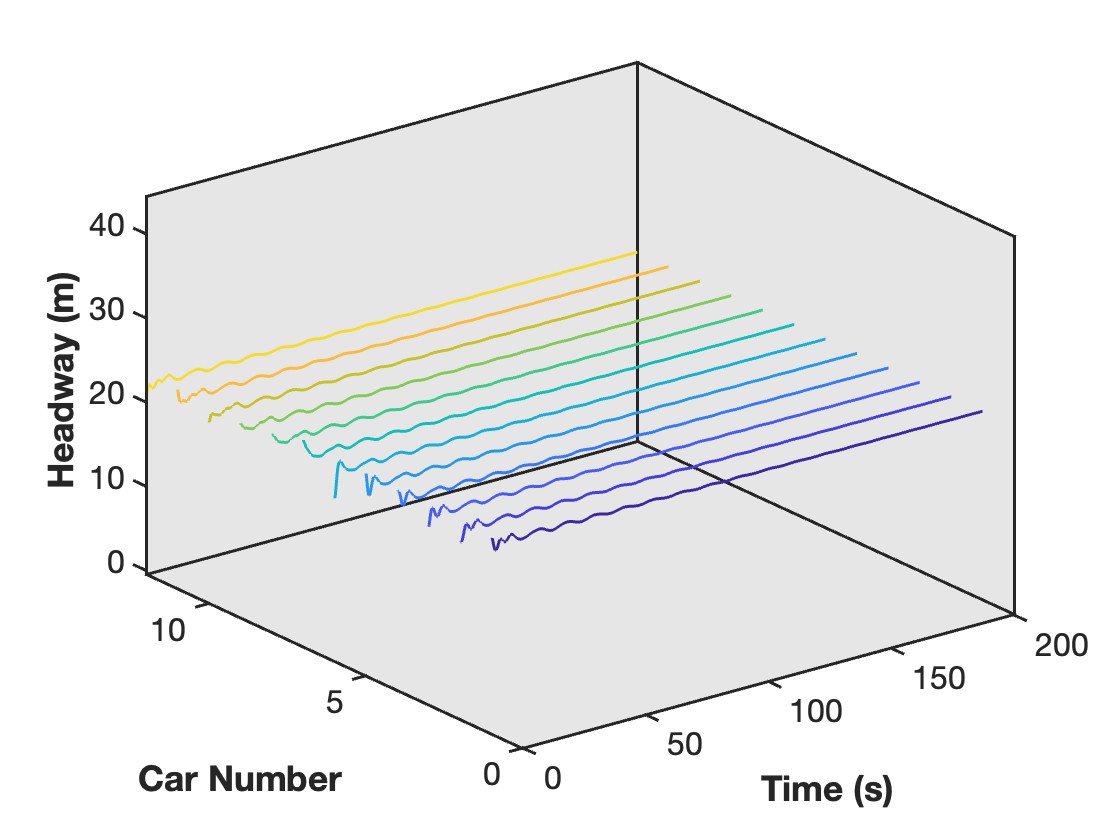}
    }
    \caption{Headway profile of OVM with sensitivity constant $a=0.4,\;0.8\;,1.6,\;2.4$}
    \label{11}
\end{figure}

\begin{figure}
    \centering
    \subfloat[$a=0.4$]{
        \includegraphics[width=0.4\textwidth]{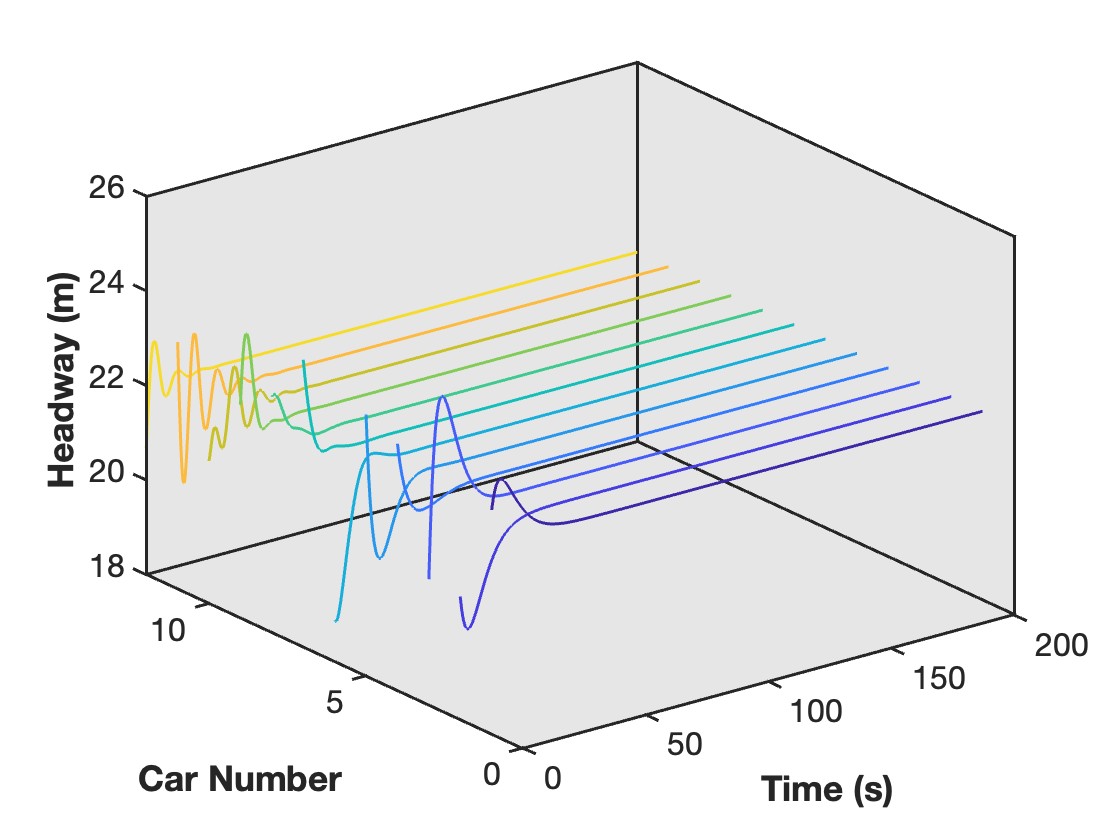}
    }
    \subfloat[$a=0.8$]{
        \includegraphics[width=0.4\textwidth]{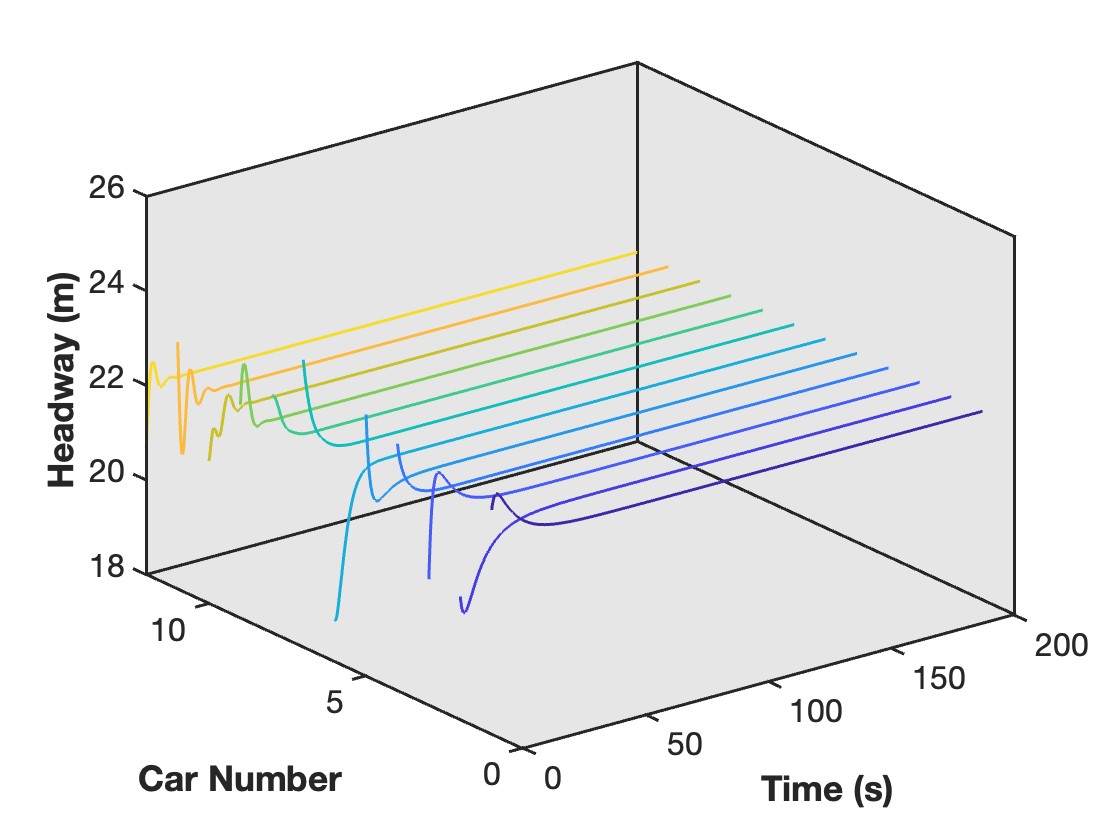}
    }
    \\
    \subfloat[$a=1.6$]{
        \includegraphics[width=0.4\textwidth]{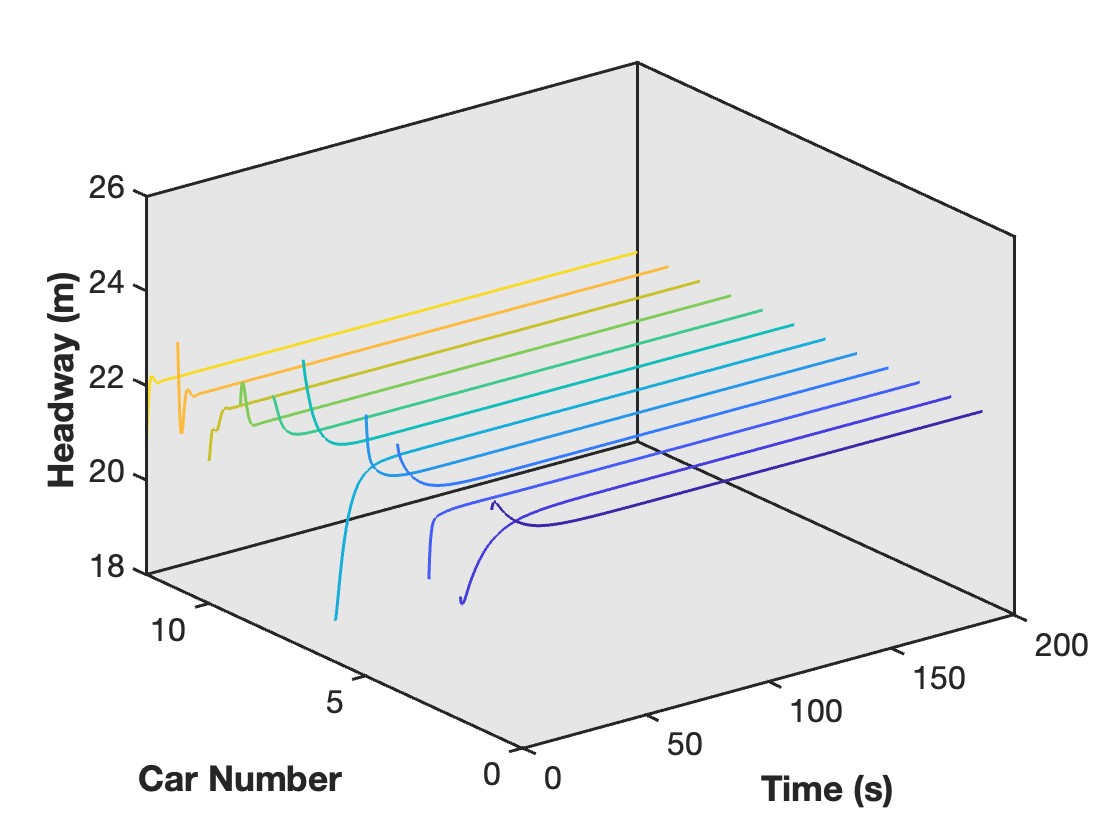}
    }
    \subfloat[$a=2.4$]{
        \includegraphics[width=0.4\textwidth]{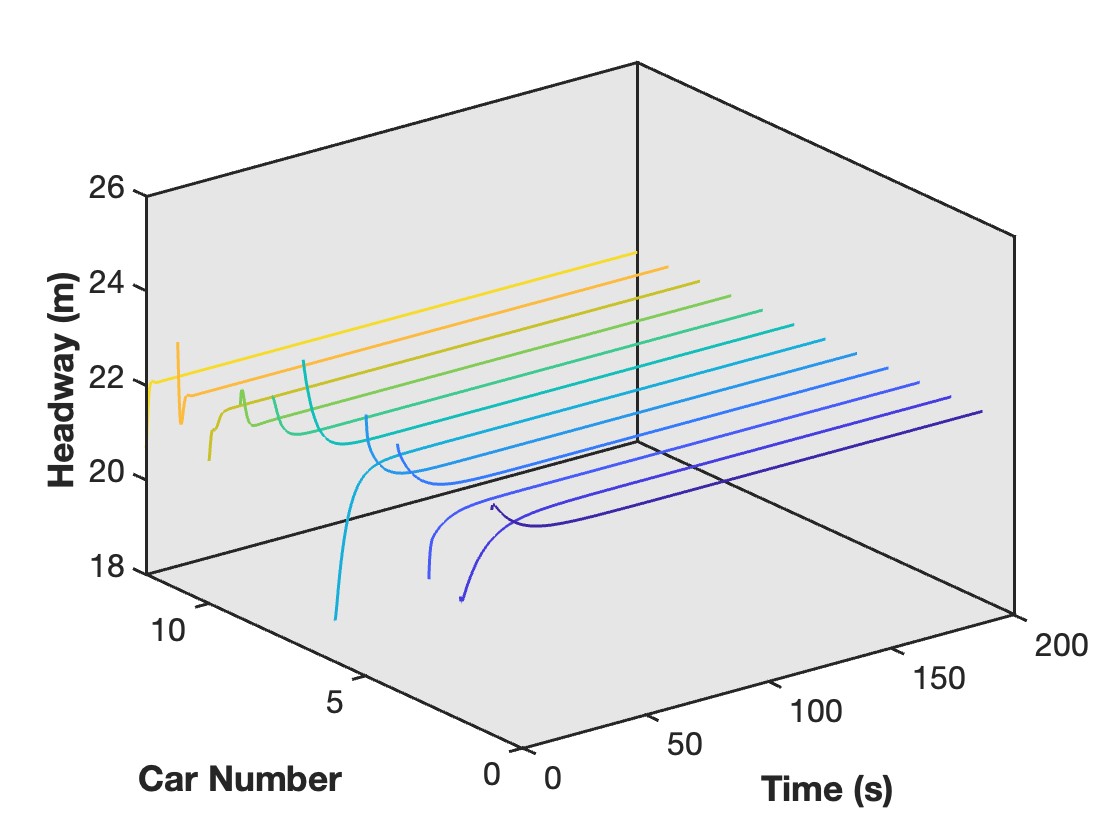}
    }
    \caption{Headway profile of P-OVM with sensitivity constant $a=0.4,\;0.8,\;1.6,\;2.4$}
    \label{12}
\end{figure}

\begin{figure}
    \centering
    \subfloat[$a=0.4$]{
        \includegraphics[width=0.4\textwidth]{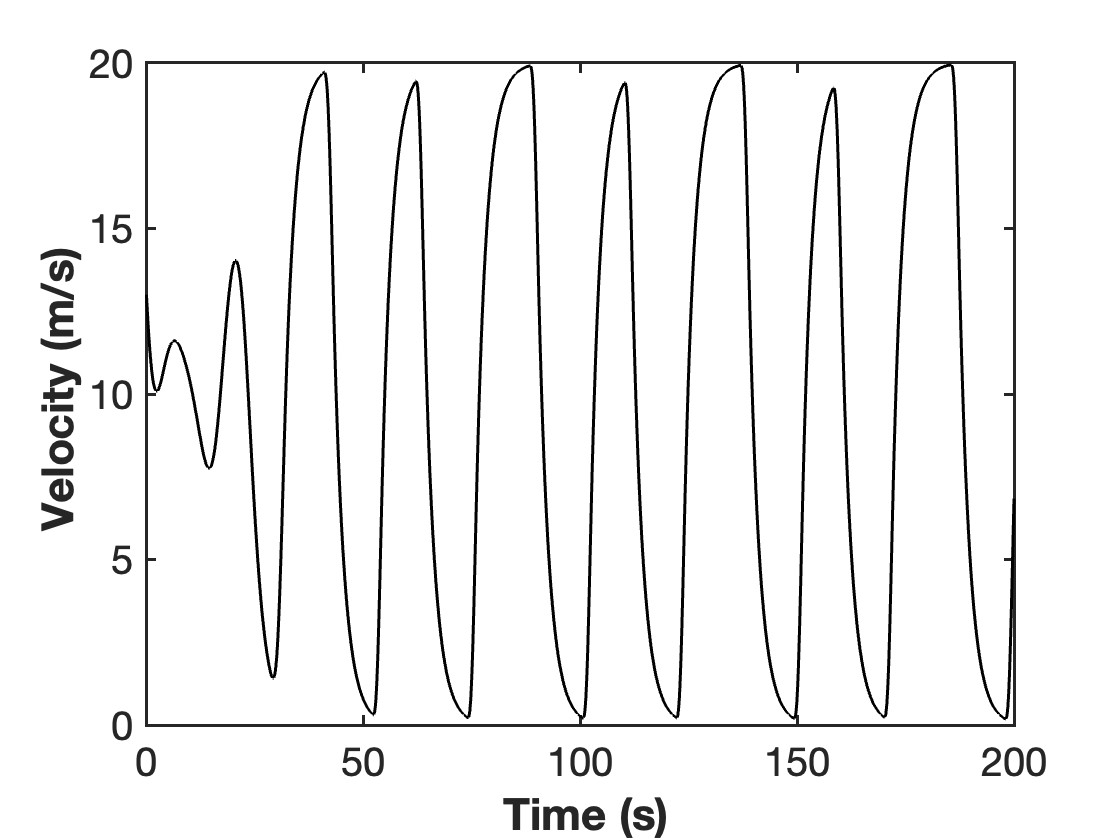}
    }
    \subfloat[$a=0.8$]{
        \includegraphics[width=0.4\textwidth]{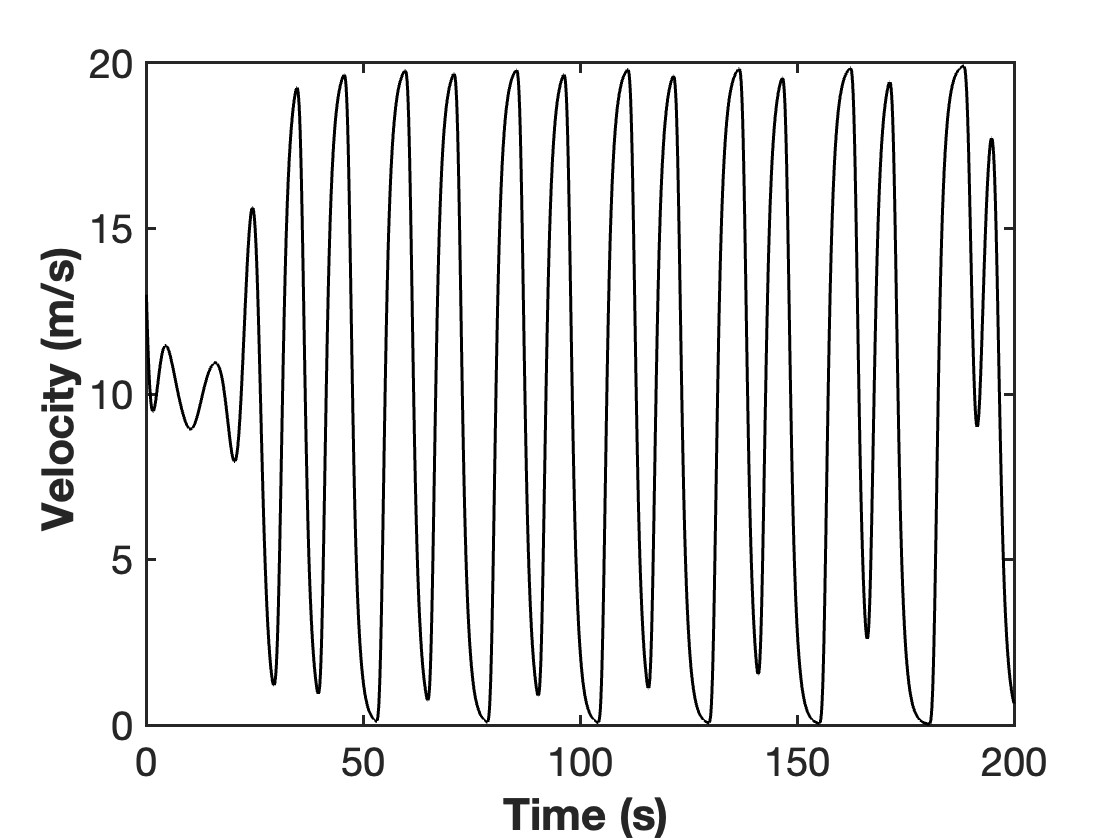}
    }
    \\
    \subfloat[$a=1.6$]{
        \includegraphics[width=0.4\textwidth]{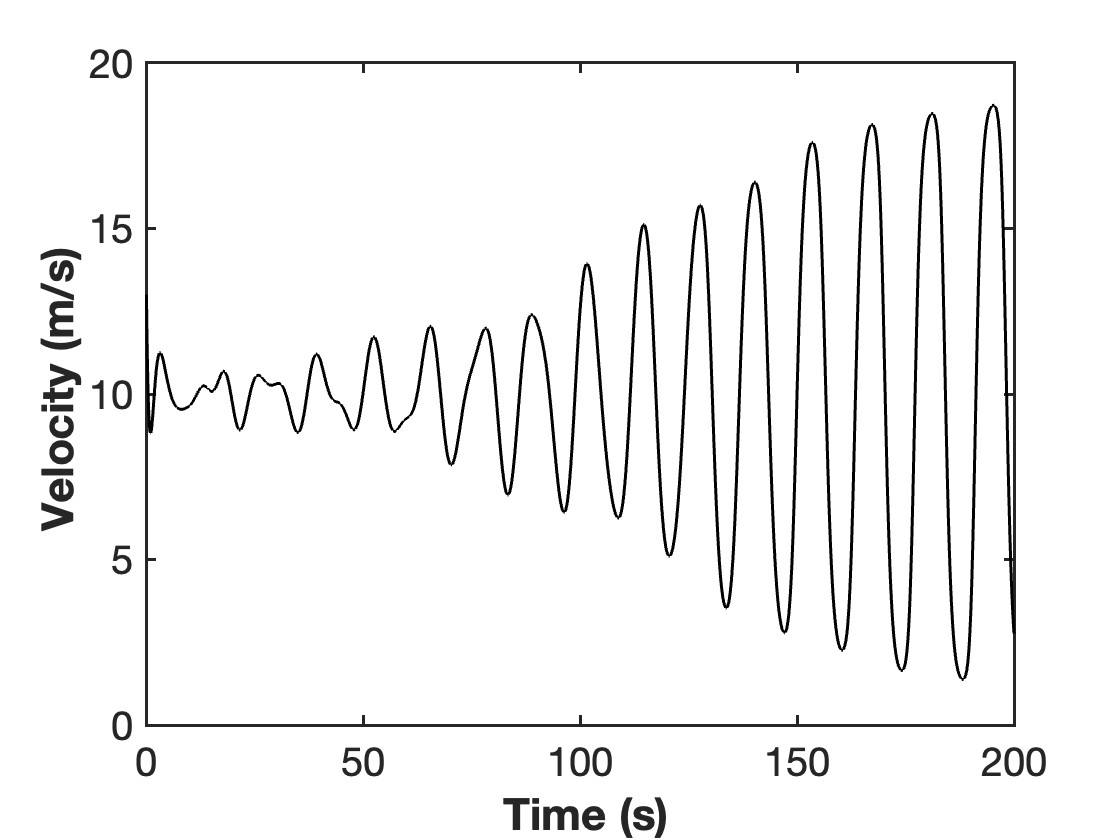}
    }
    \subfloat[$a=2.4$]{
        \includegraphics[width=0.4\textwidth]{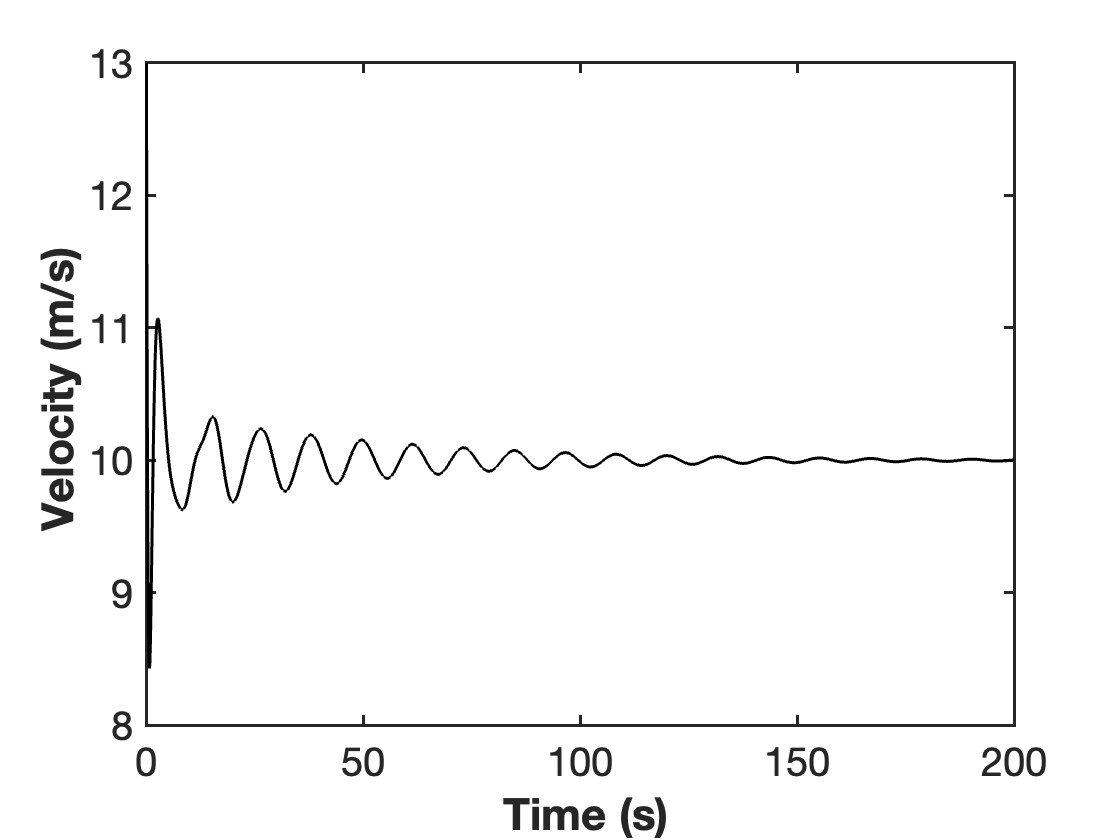}
    }
    \caption{Velocity profile of $6$th vehicle of OVM with sensitivity constant $a=0.4,\;0.8,\;1.6,\;2.4$}
    \label{13}
\end{figure}

\begin{figure}
    \centering
    \subfloat[$a=0.4$]{
        \includegraphics[width=0.4\textwidth]{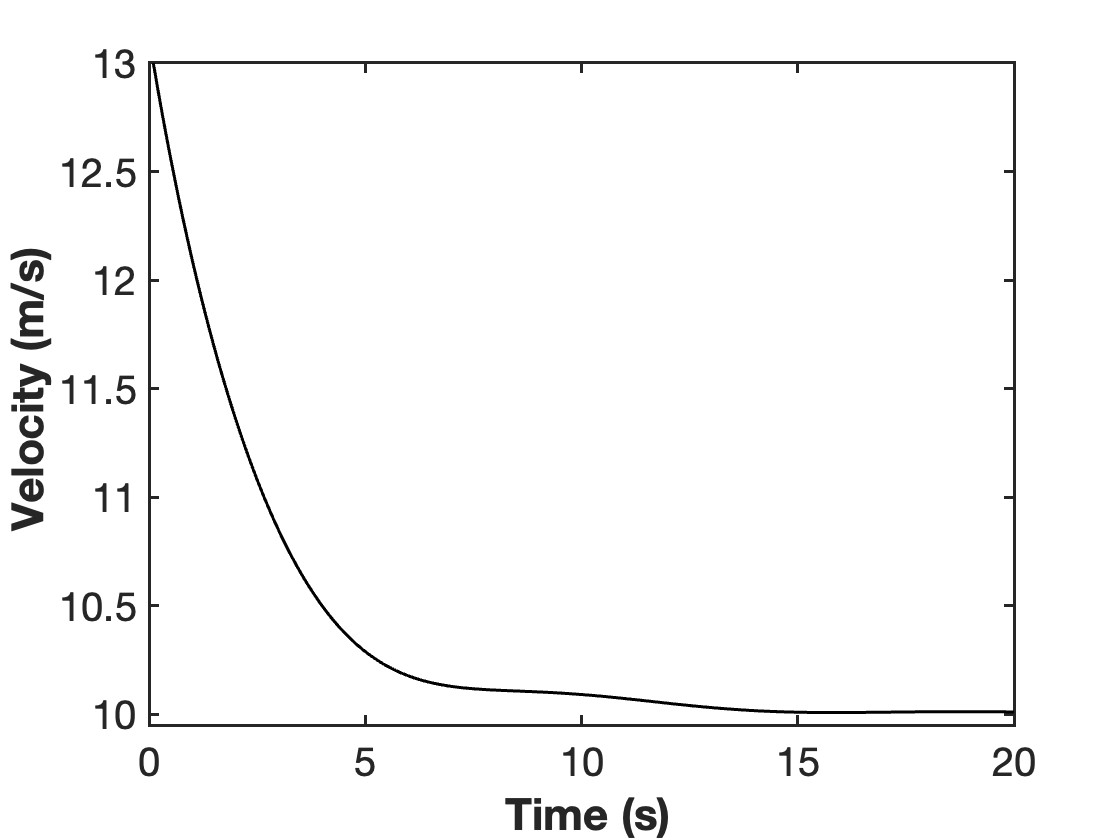}
    }
    \subfloat[$a=0.8$]{
        \includegraphics[width=0.4\textwidth]{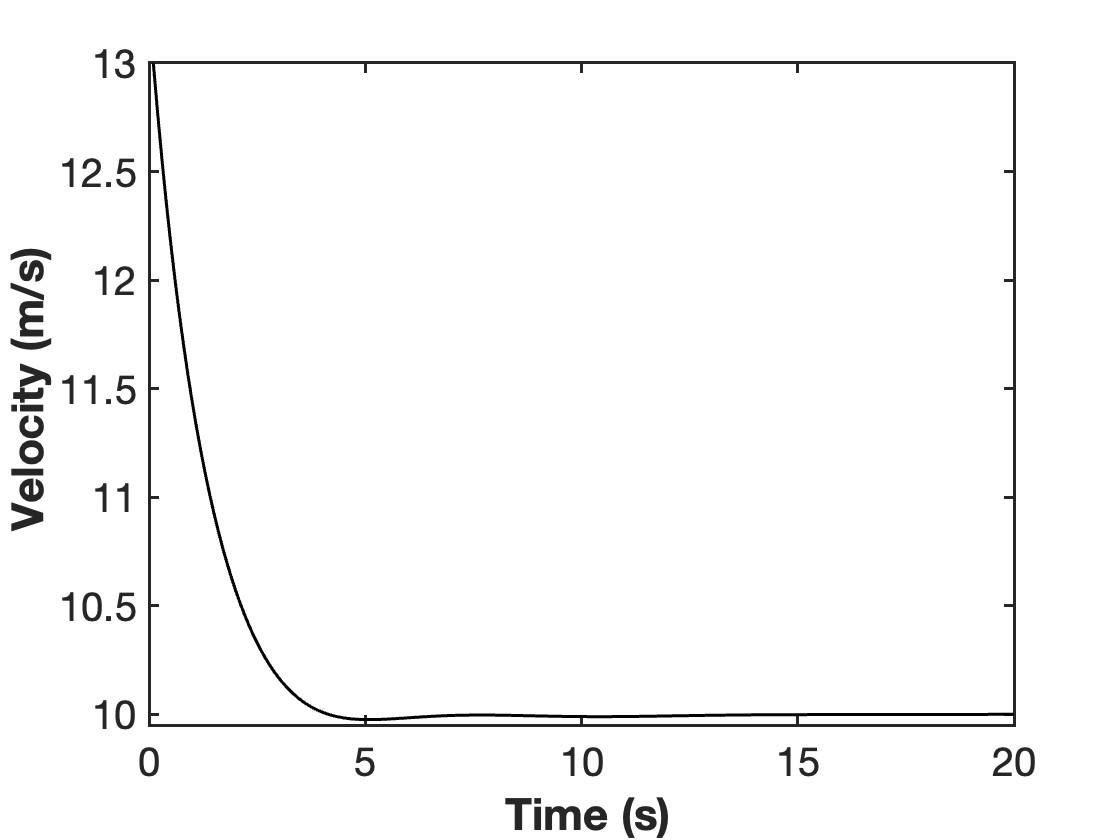}
    }
    \\
    \subfloat[$a=1.6$]{
        \includegraphics[width=0.4\textwidth]{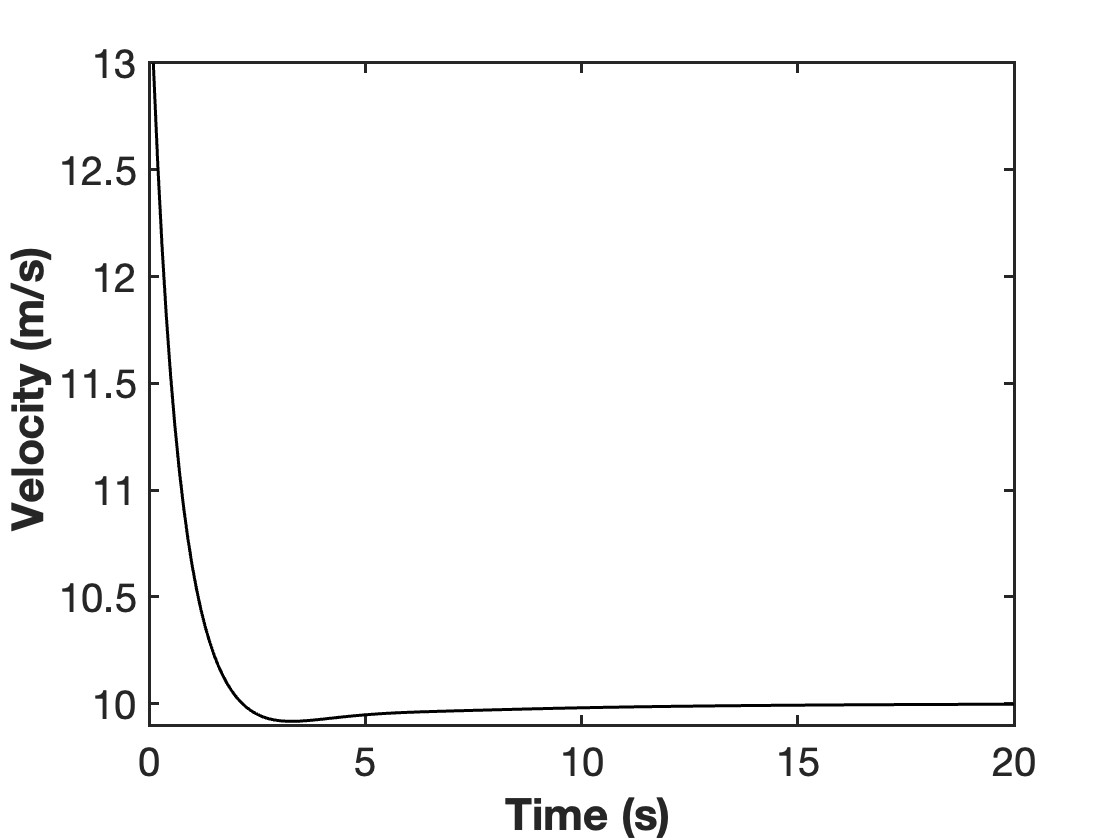}
    }
    \subfloat[$a=2.4$]{
        \includegraphics[width=0.4\textwidth]{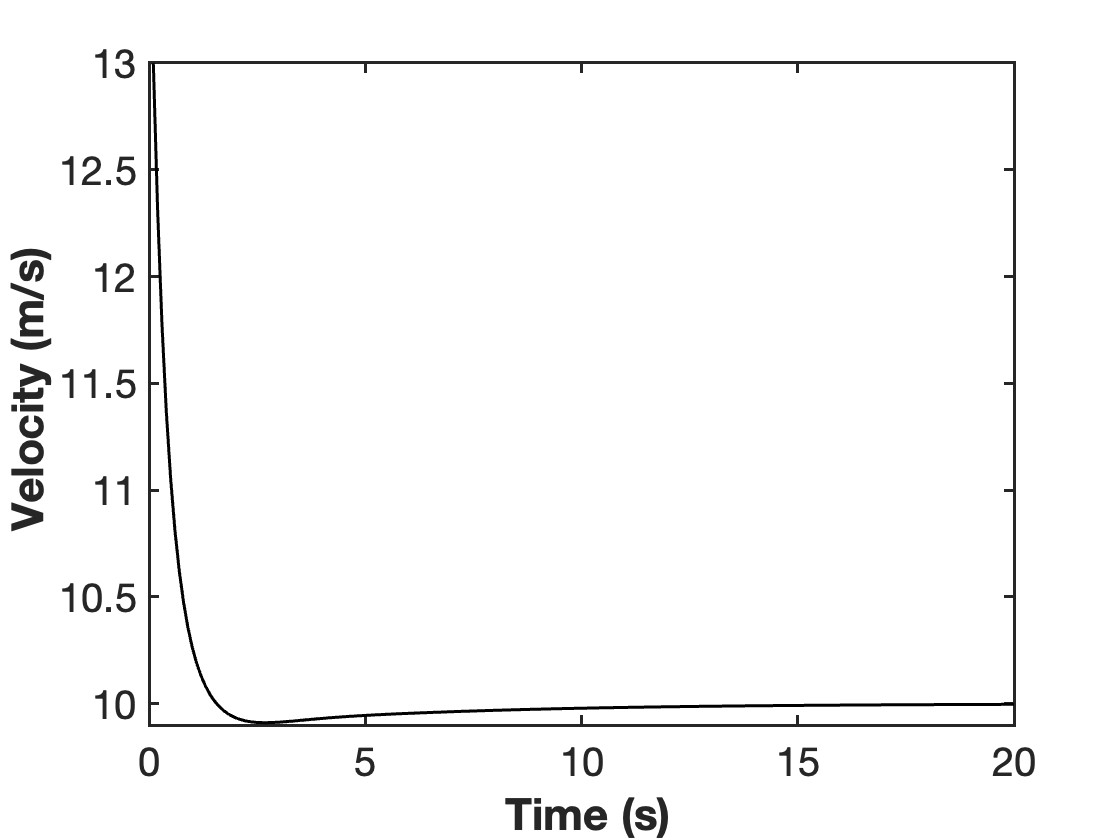}
    }
    \caption{Velocity profile of $6$th vehicle of P-OVM with sensitivity constant $a=0.4,\;0.8,\;1.6,\;2.4$}
    \label{14}
\end{figure}
% interpret simulation results
The simulation results show that the OVM is unstable for $a=0.4,\;0.8,\;1.6$. At $a=0.4$, negative headway are generated, indicating a risk of collisions due to low sensitivity. For $a=2.4$ although the OVM is stable, the initial acceleration for the selected vehicle is unrealistically high. In contrast, the P-OVM remains stable for all selected values of $a$ the model is stable and $a$ with sensitivity primarily affecting the convergence speed. For greater values of $a$, traffic converge to the equilibrium state faster. Yet when $a$ is sufficiently large, further improvements in convergence become negligible, as the original OVM is already stable.

\textbf{Simulation 1.2}: Test of some T-OVM parameters
% T-OVM explanation

To test the performance of T-OVM, four cases of sensitivity pair: $(a,b)=(0.5,0.1),$ $(0.1,0.5),$ $(1,0.2),$ $(0.6,0.6)$ are tested under the same initial conditions. The extreme cases where $a=0$ or $b=0$ are not considered as they essentially reduce to OVM or P-OVM, respectively. Simulation results are in Figures \ref{21}-\ref{22}. Figure \ref{21} is the headway profile of all 12 vehicles in 3-D plot with different pairs of sensitivity constants, and Figure \ref{22} is the velocity profile of the $6$th vehicle.

%Plots of this simulation% 
\begin{figure}
    \centering
    \subfloat[$(a,b)=(0.5,0.1)$]{
        \includegraphics[width=0.4\textwidth]{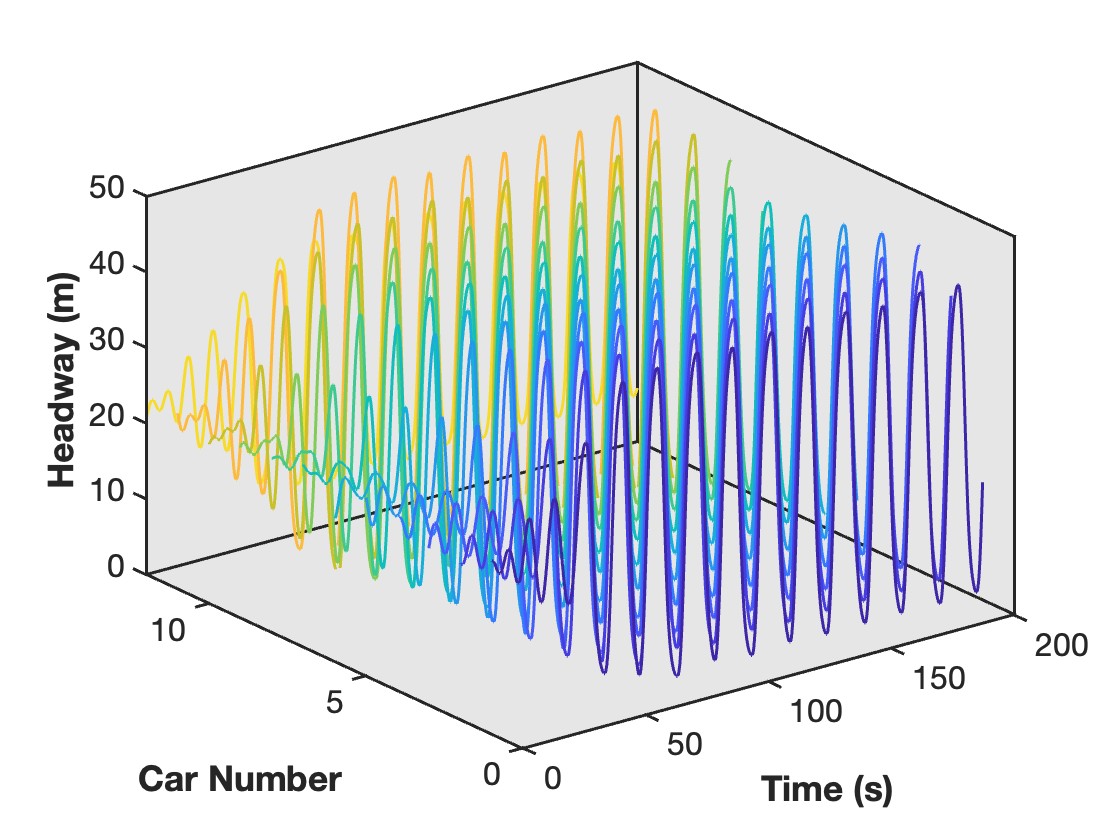}
    }
    \subfloat[$(a,b)=(0.1,0.5)$]{
        \includegraphics[width=0.4\textwidth]{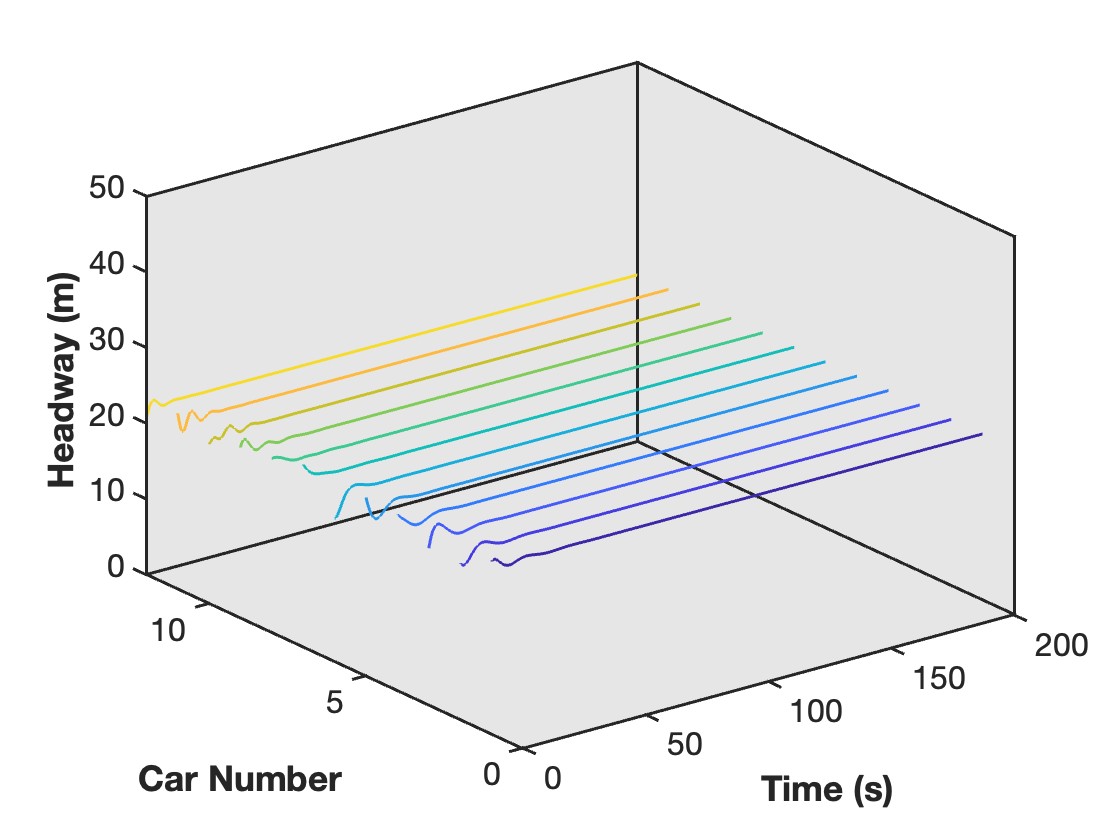}
    }
    \\
    \subfloat[$(a,b)=(1,0.2)$]{
        \includegraphics[width=0.4\textwidth]{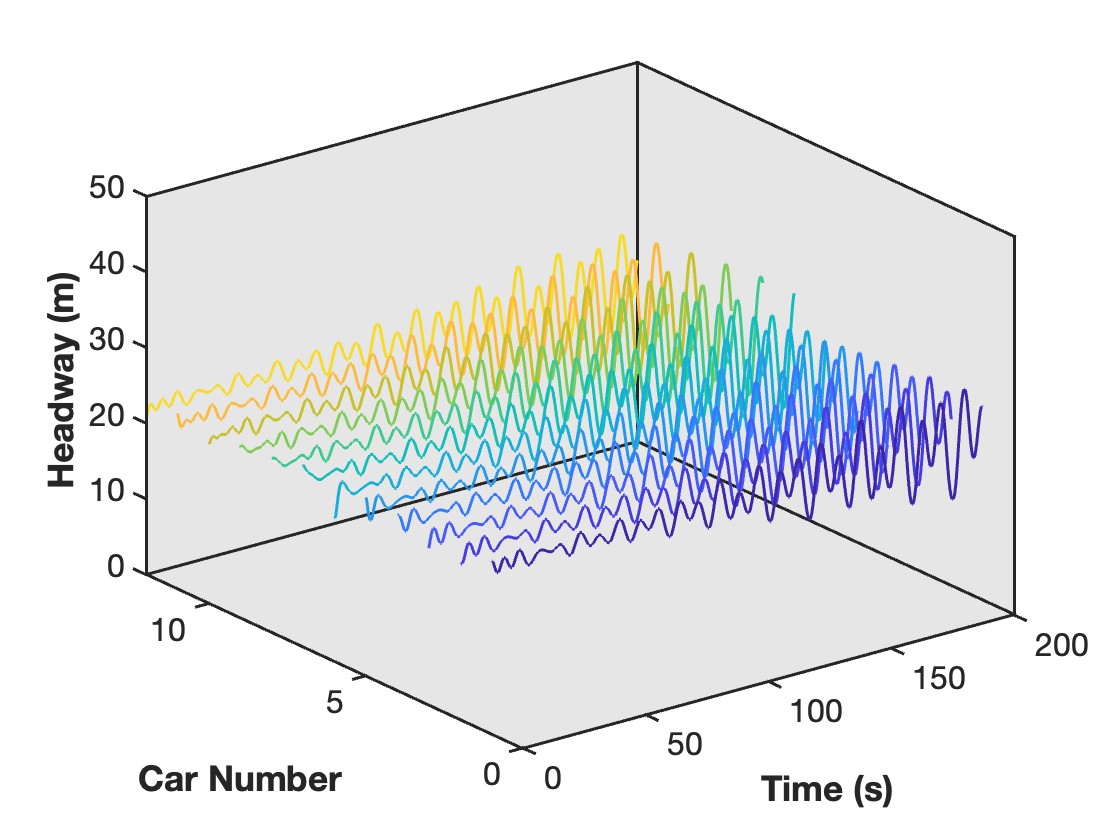}
    }
    \subfloat[$(a,b)=(0.6,0.6)$]{
        \includegraphics[width=0.4\textwidth]{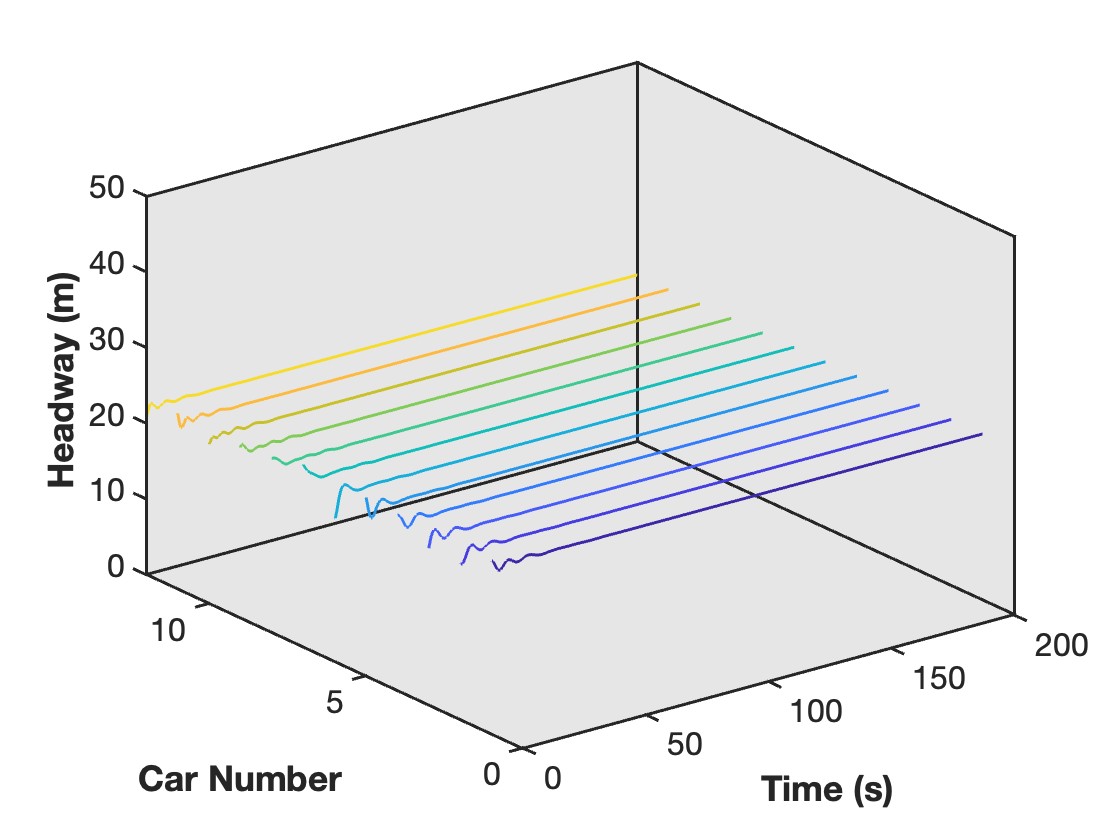}
    }
    \caption{Headway profile of T-OVM with sensitivity constant $(a,b)=(0.5,0.1),(0.1,0.5),(1,0.2),(0.6,0.6)$}
    \label{21}
\end{figure}

\begin{figure}
    \centering
    \subfloat[$(a,b)=(0.5,0.1)$]{
        \includegraphics[width=0.4\textwidth]{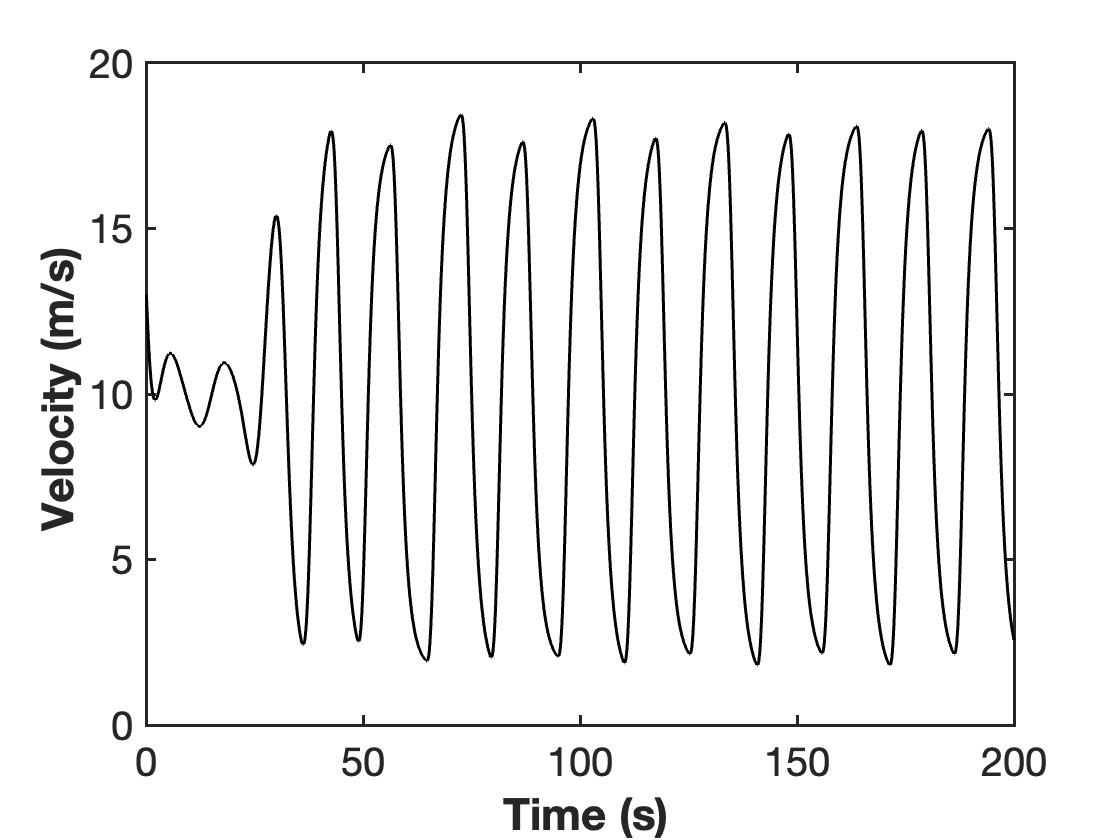}
    }
    \subfloat[$(a,b)=(0.1,0.5)$]{
        \includegraphics[width=0.4\textwidth]{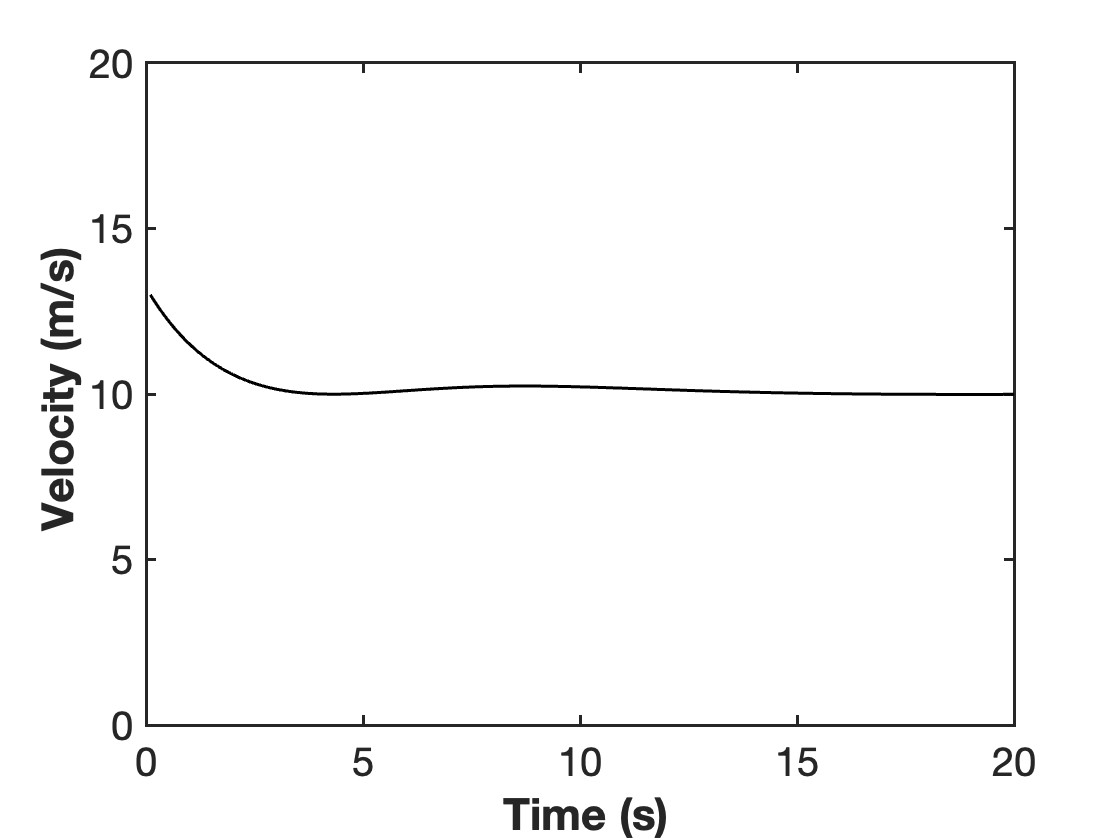}
    }
    \\
    \subfloat[$(a,b)=(1,0.2)$]{
        \includegraphics[width=0.4\textwidth]{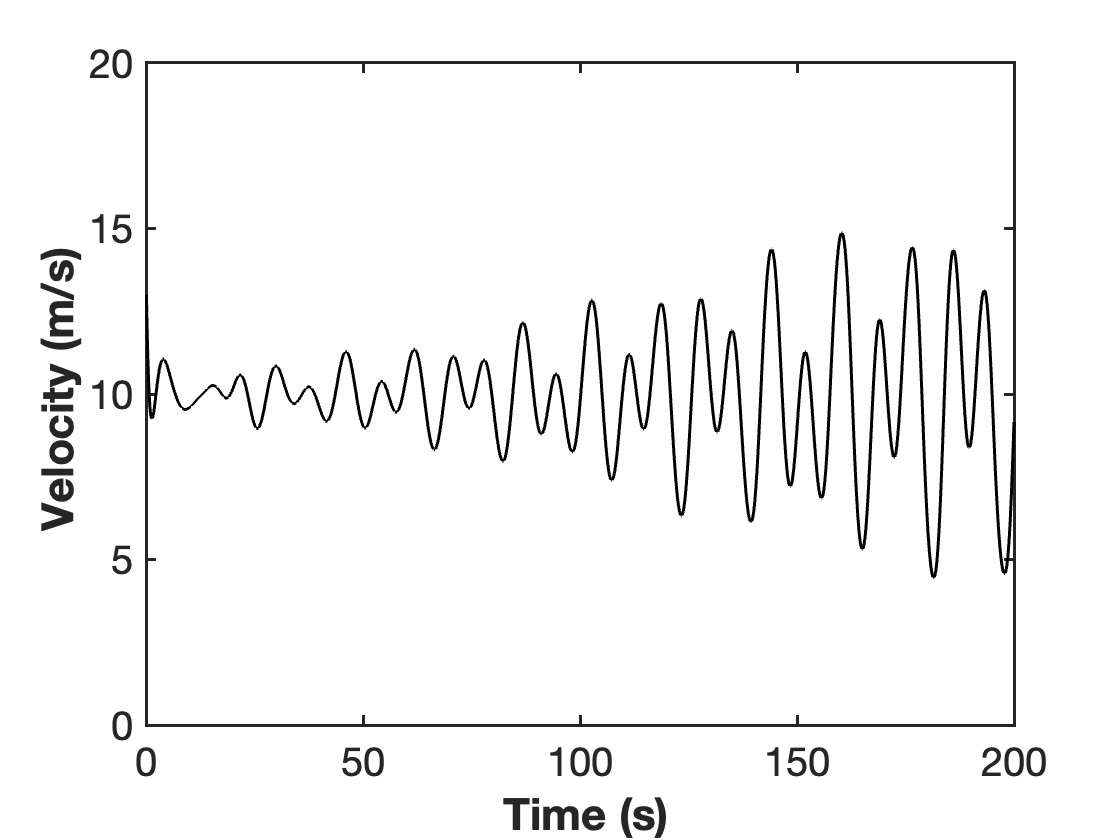}
    }
    \subfloat[$(a,b)=(0.6,0.6)$]{
        \includegraphics[width=0.4\textwidth]{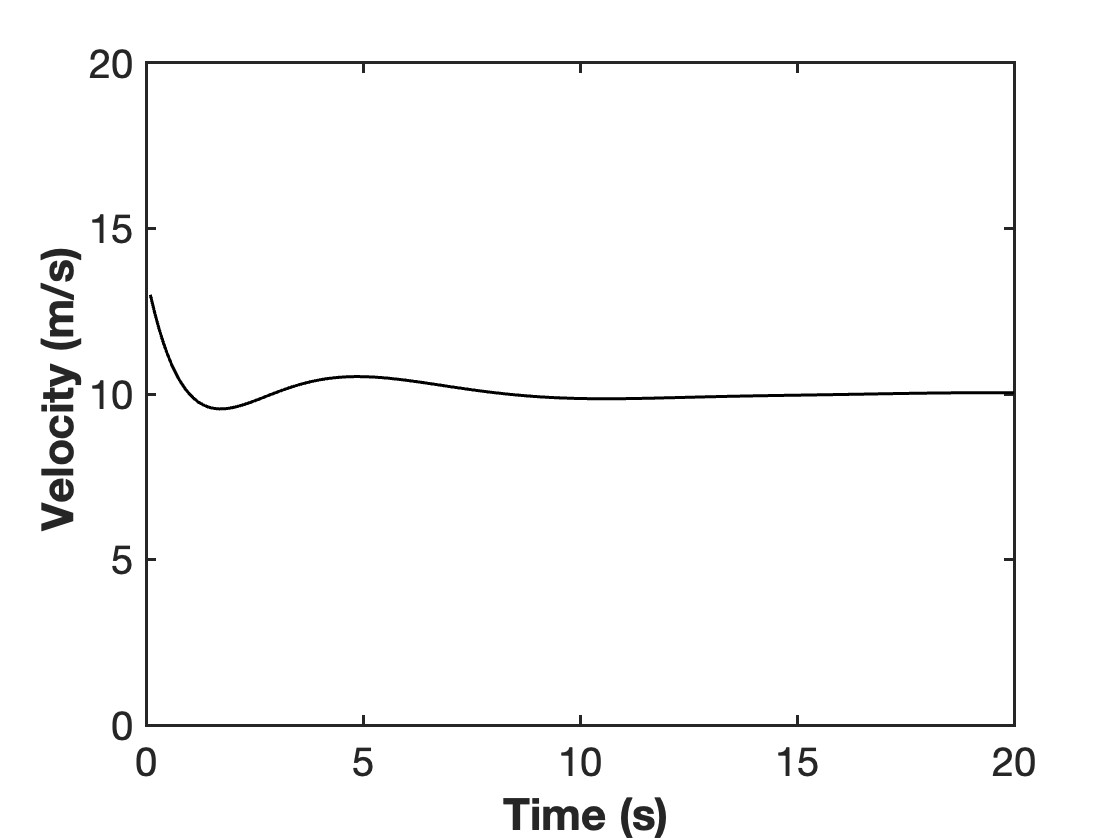}
    }
    \caption{Velocity profile of $6$th vehicle of T-OVM with sensitivity constant $(a,b)=(0.5,0.1),(0.1,0.5),(1,0.2),(0.6,0.6)$}
    \label{22}
\end{figure}

% interpret simulation results
The simulation results show that as the percentage of sensitivity to the leading vehicle increases, the traffic stream becomes more stable. Additionally, for higher total sensitivity $(a+b)$ lower percentages of leading vehicle sensitivity is required to achieve string stability, which aligns with the theoretical results (see figure \ref{s1}).

\textbf{Simulation 1.3}: Comparison of the T-OVM with the two cars ahead following OVM
% T-OVM explanation

To show the effectiveness of following the platoon leader, we compare the T-OVM \eqref{tran} with the front multi-following OVM (F-OVM) of two cars in \cite{lenz1999multi}:
\begin{equation}
    \ddot{x}_i=a(V(x_{i+1}-x_i)-\dot{x}_i)+b\left(V\left(\frac{x_{i+2}-x_i}{2}\right)-\dot{x}_i\right),
    \label{OVM2front}
\end{equation}
where the sensitivity $b$ is for the second vehicle in front. Two cases of sensitivity pair: $(a,b)=(0.8, 0.4),(0.2, 0.4)$ are tested. Simulation results are in Figures \ref{add1}, \ref{add2}. Figure \ref{add1} is the headway profile of all 12 vehicles in 3-D plot with T-OVM and F-OVM and Figure \ref{add2} is the velocity profile of the $6$th vehicle.

%Plots of this simulation% 
\begin{figure}
    \centering
    \subfloat[T-OVM with $(a,b)=(0.8, 0.4)$]{
        \includegraphics[width=0.4\textwidth]{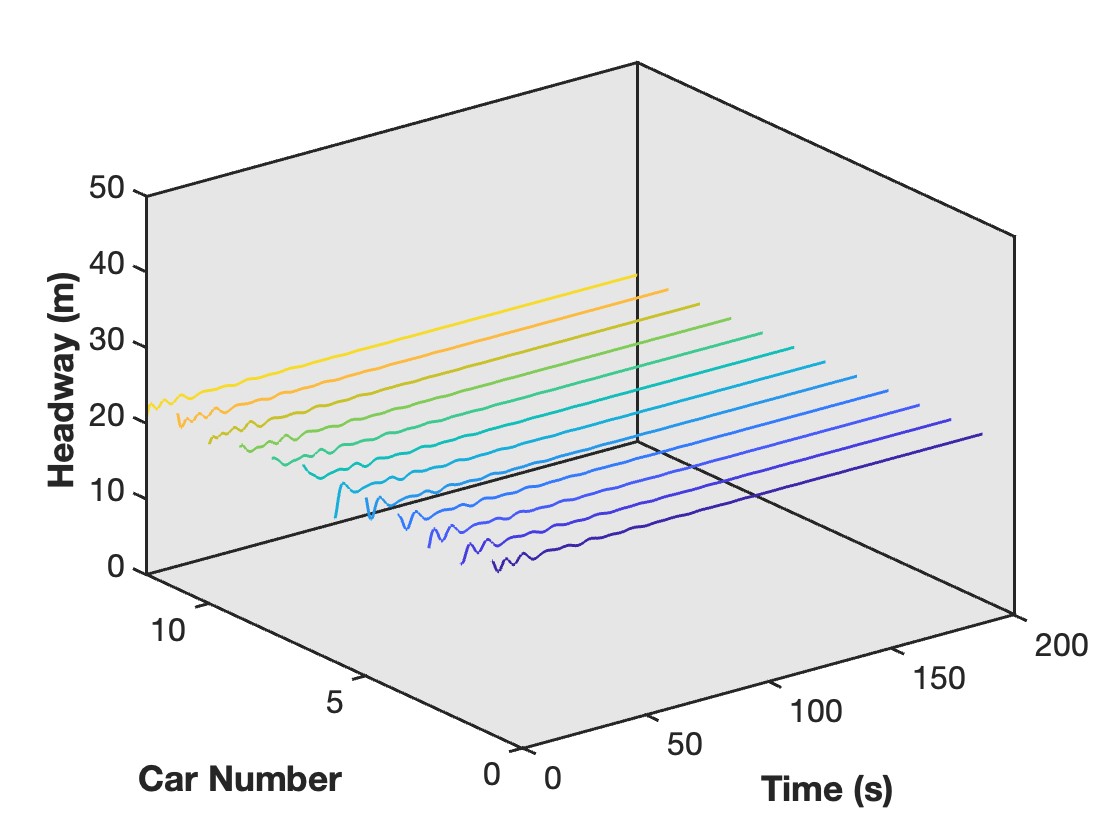}
    }
    \subfloat[F-OVM with $(a,b)=(0.8, 0.4)$]{
        \includegraphics[width=0.4\textwidth]{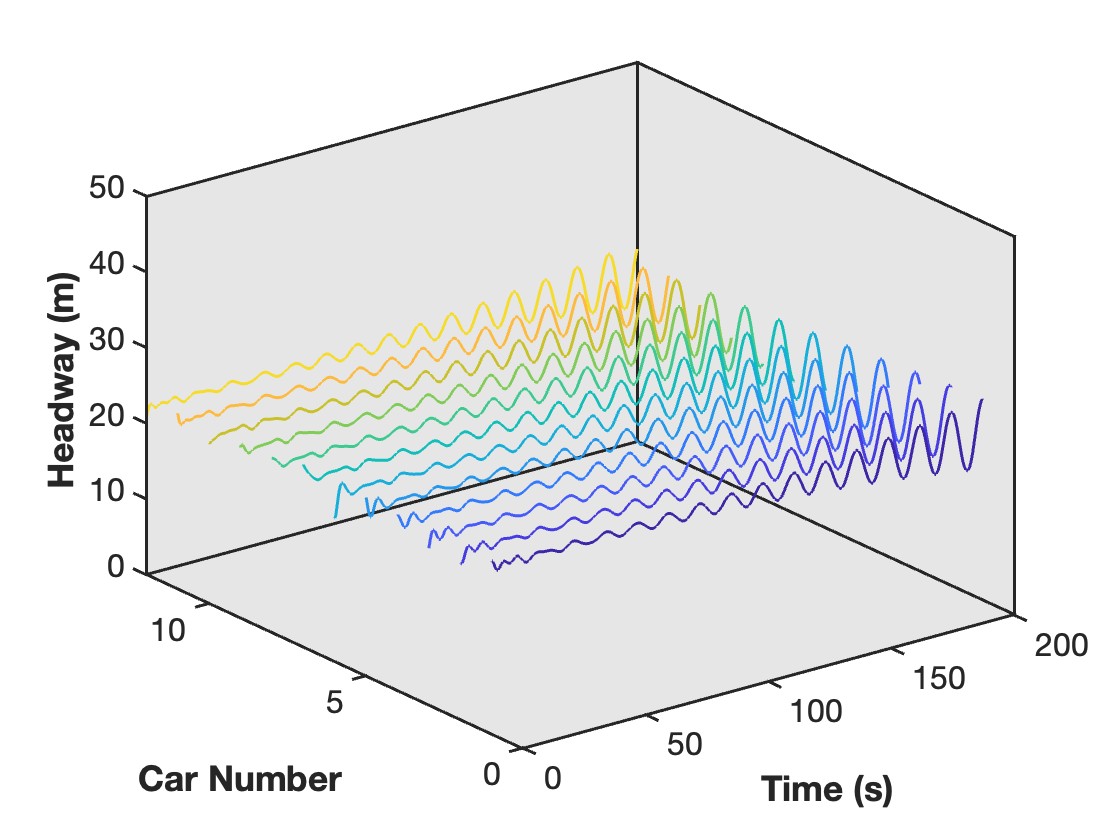}
    }
    \\
    \subfloat[T-OVM with $(a,b)=(0.2, 0.4)$]{
        \includegraphics[width=0.4\textwidth]{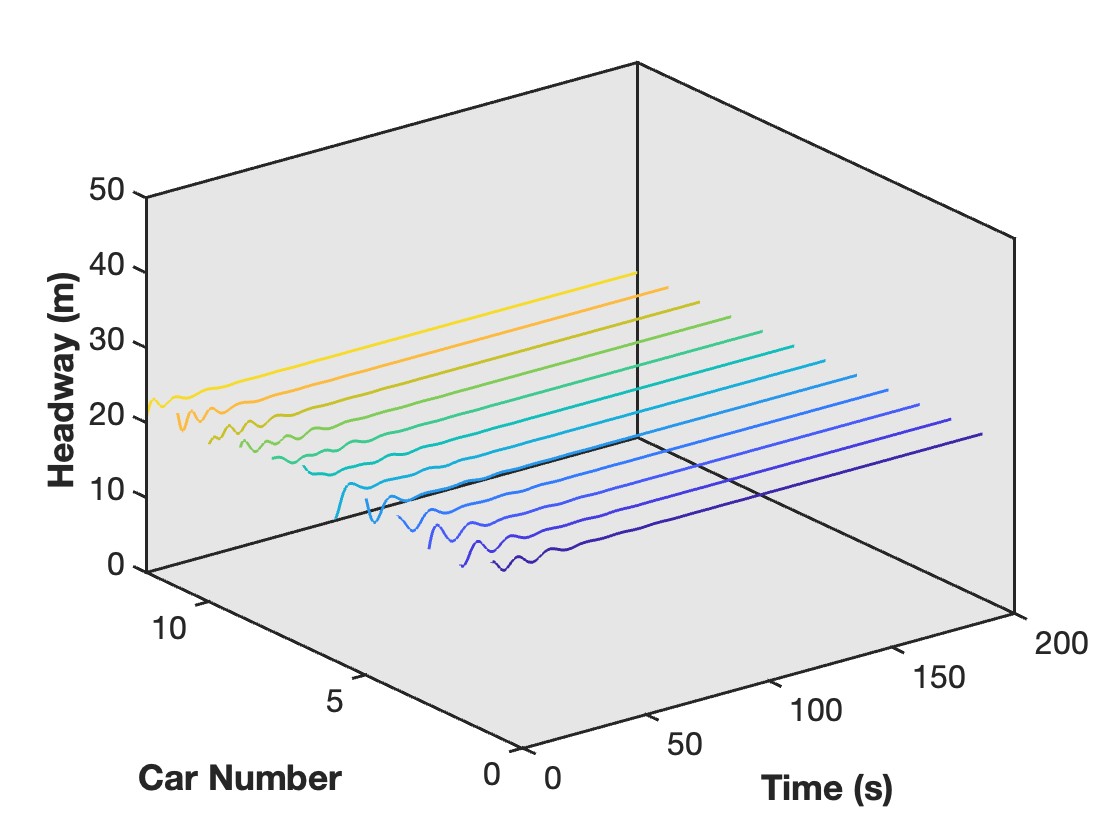}
    }
    \subfloat[F-OVM with $(a,b)=(0.2, 0.4)$]{
        \includegraphics[width=0.4\textwidth]{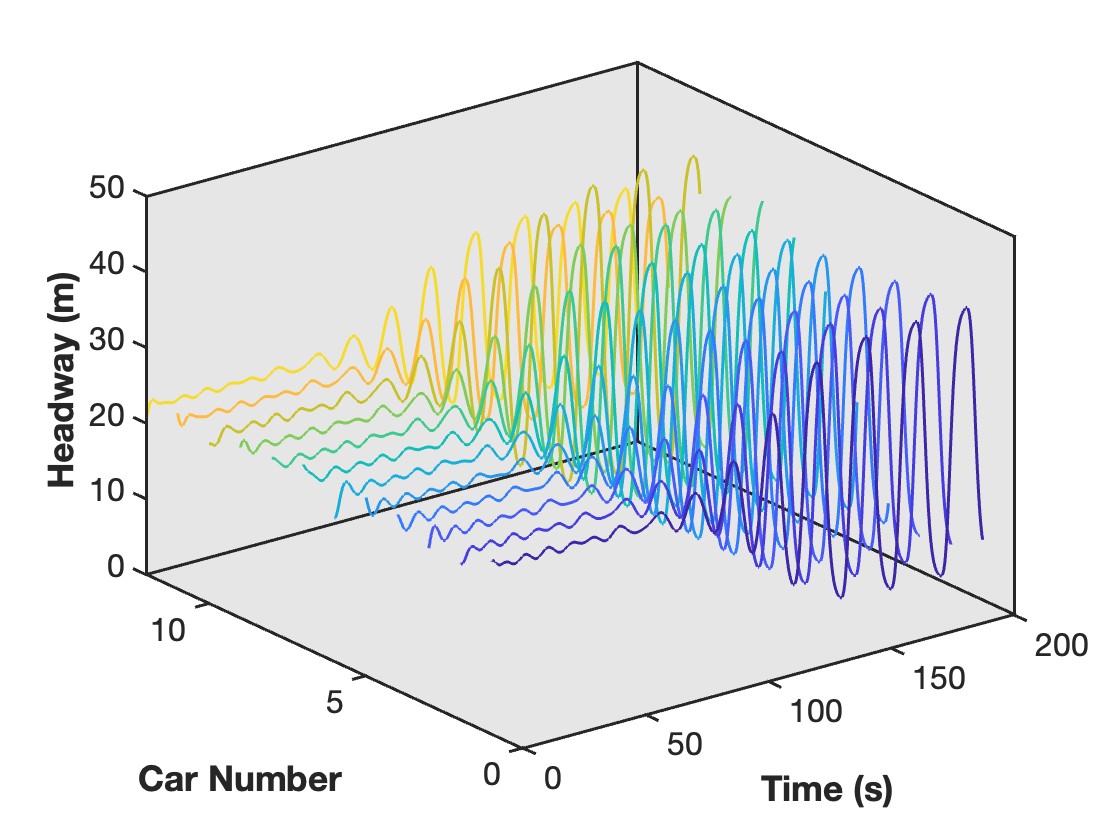}
    }
    \caption{Headway profile of T-OVM and F-OVM with sensitivity constant $(a,b)=(0.8, 0.4),(0.2, 0.4)$}
    \label{add1}
\end{figure}

\begin{figure}
    \centering
    \subfloat[T-OVM with $(a,b)=(0.8, 0.4)$]{
        \includegraphics[width=0.4\textwidth]{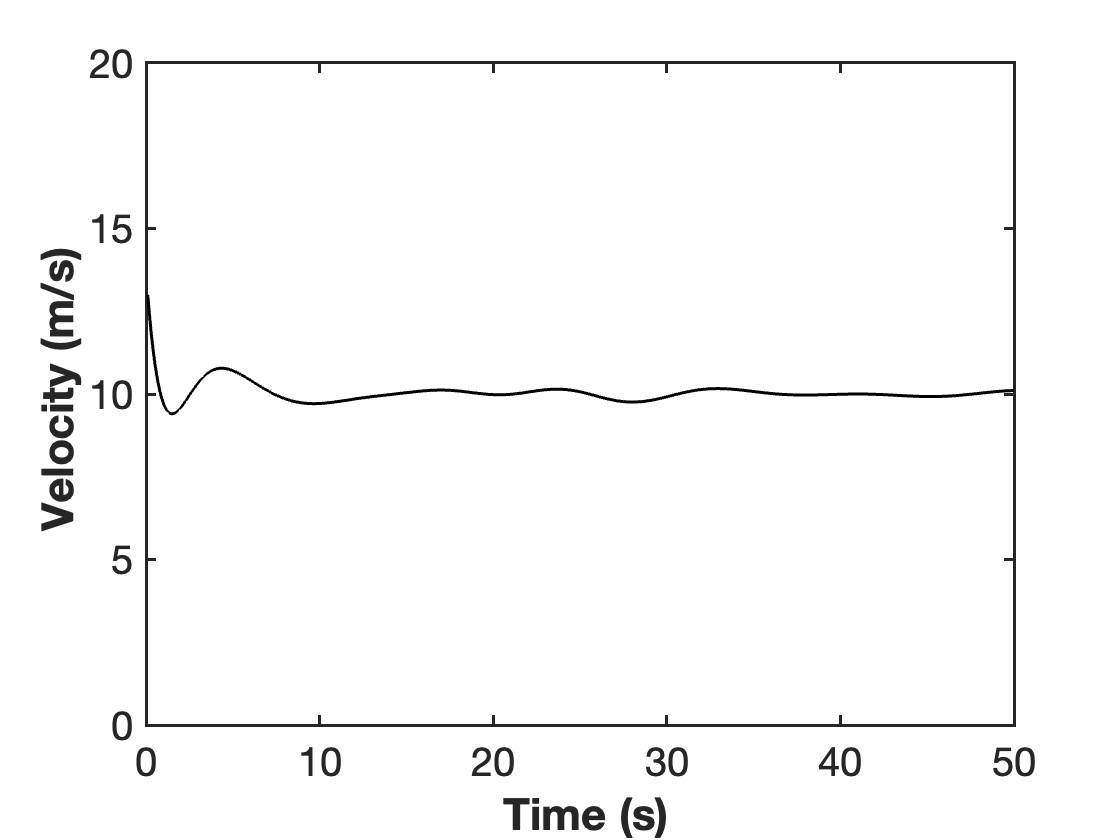}
    }
    \subfloat[F-OVM with $(a,b)=(0.8, 0.4)$]{
        \includegraphics[width=0.4\textwidth]{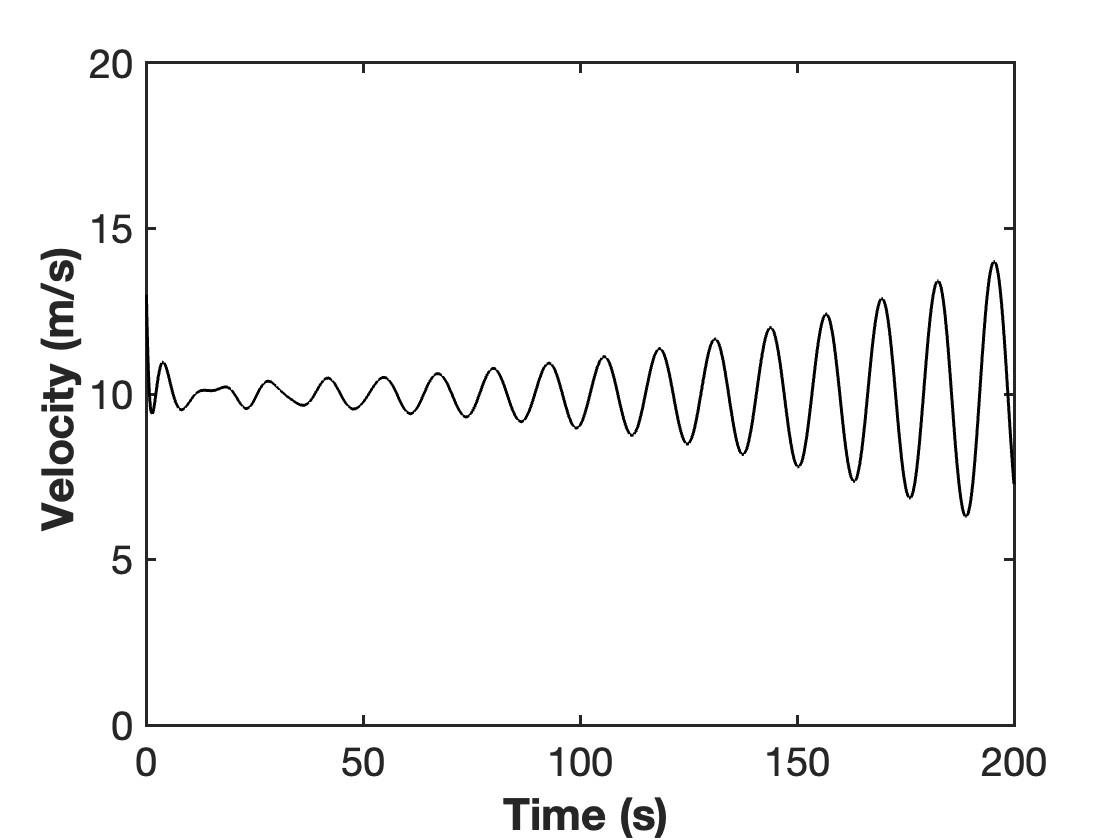}
    }
    \\
    \subfloat[T-OVM with $(a,b)=(0.2, 0.4)$]{
        \includegraphics[width=0.4\textwidth]{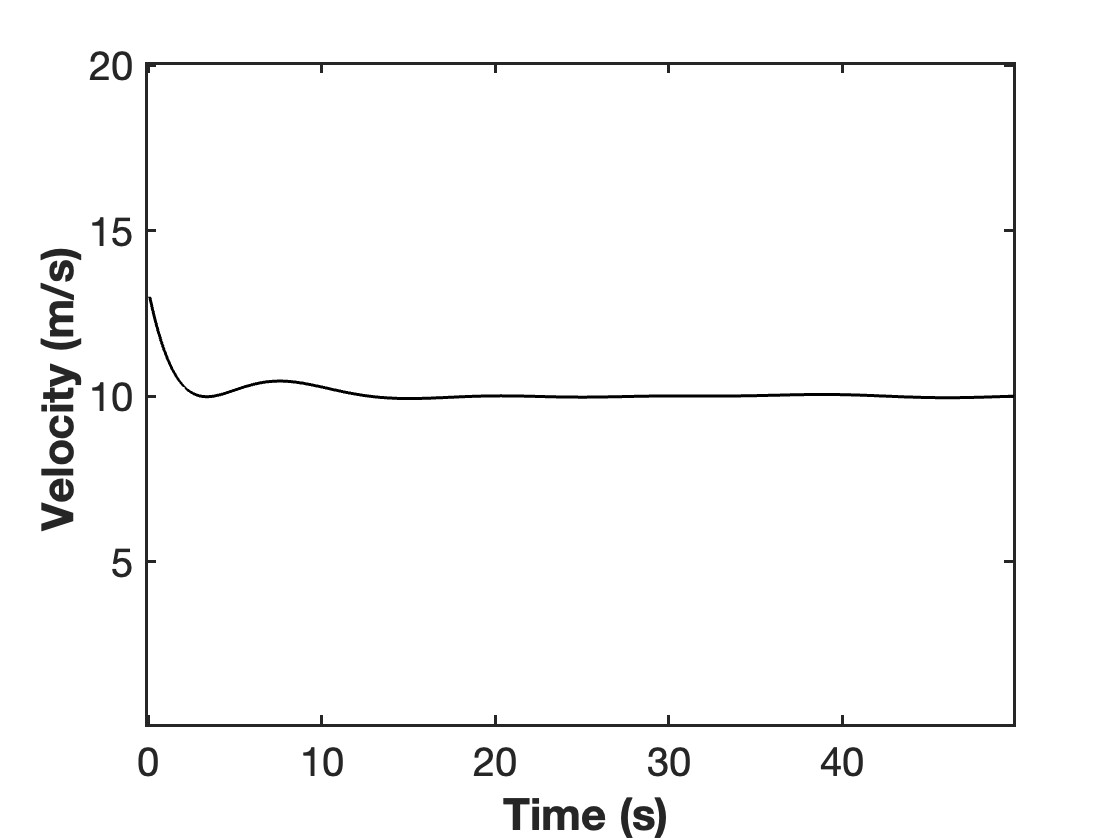}
    }
    \subfloat[F-OVM with $(a,b)=(0.2, 0.4)$]{
        \includegraphics[width=0.4\textwidth]{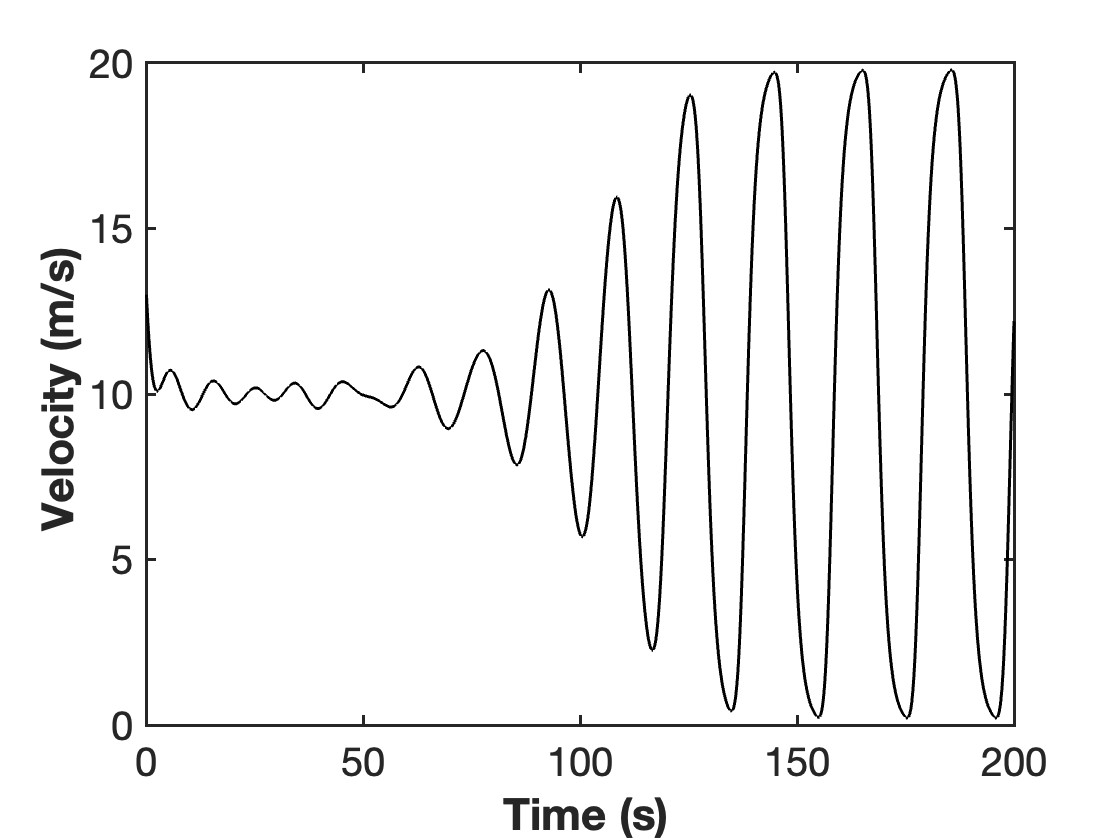}
    }
    \caption{Velocity profile of $6$th vehicle of T-OVM and F-OVM with sensitivity constant $(a,b)=(0.8, 0.4),(0.2, 0.4)$}
    \label{add2}
\end{figure}

From the simulation results we can observe that for the same setting of sensitivity parameters, T-OVM can suppress the disturbance and the platoon remains stable while F-OVM cannot. Furthermore, the stabilizing effect of the T-OVM becomes more prominent when the following vehicles respond more strongly to the leader in the T-OVM, while the destabilizing effect in the F-OVM becomes stronger when the following vehicle respond more strongly to the second vehicle ahead. These findings demonstrate the superiority of applying T-OVM in stabilizing CAV platoons.

\subsection{Infinite road with periodic disturbance}

In addition to the initial disturbance, it is also valuable to examine whether platoon control can enhance stability under periodic disturbances. In this subsection we consider $N=10$ vehicles travelling on an infinite road with a free flow speed of $v_{max}=30$m/s and a uniform length of $l=5$m. The optimal velocity function $V_I(h)$ for the infinite road simulation is equivalent to a triangular fundamental diagram:
\begin{equation}
    V_I(h)=\begin{cases}
        v_{\max}, & \text{if } \rho(h)\leq \rho_c; \\
        \frac{v_{\max}*\rho_c(\rho(h)-\rho_{\max})}{\rho(h)(\rho_c-\rho_{\max})}, & \text{if } \rho_c\leq \rho(h) \leq \rho_{\max};\\
        0, & \text{if } \rho\geq \rho_{\max},
    \end{cases}
\end{equation}
where $\rho(h)=l/h$ is the occupancy of vehicles on the road where $\rho=0$ indicates an empty road, and $\rho=1$ represents full vehicle occupancy. $\rho_c=5/37$ is the critical occupancy where flow is maximized and $\rho_{\max}=5/7$ is the jam occupancy. Figure \ref{f2} is a plot of the optimal velocity function and corresponding triangular fundamental diagram.
% plot of triangular FD 
\begin{figure}
    \centering
\subfloat[Optimal velocity function]{
    \includegraphics[width=0.4\textwidth]{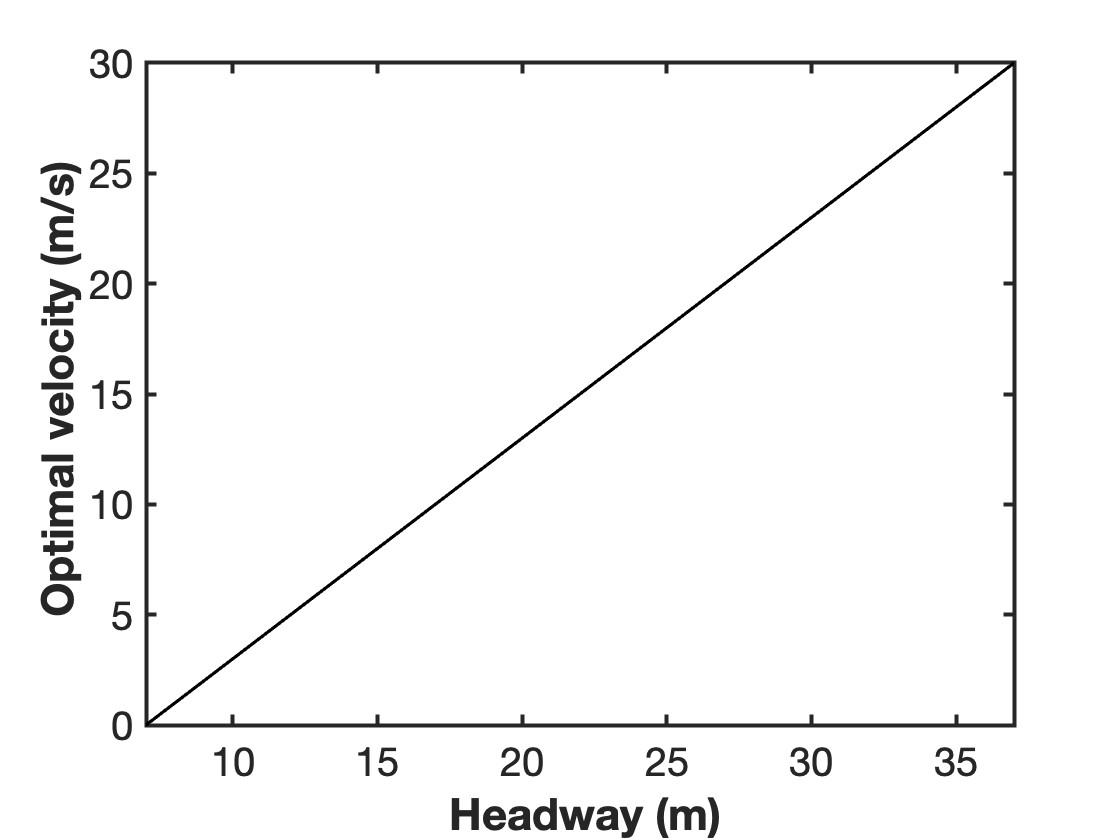}
}
\subfloat[Triangular fundamental diagram]{
    \includegraphics[width=0.4\textwidth]{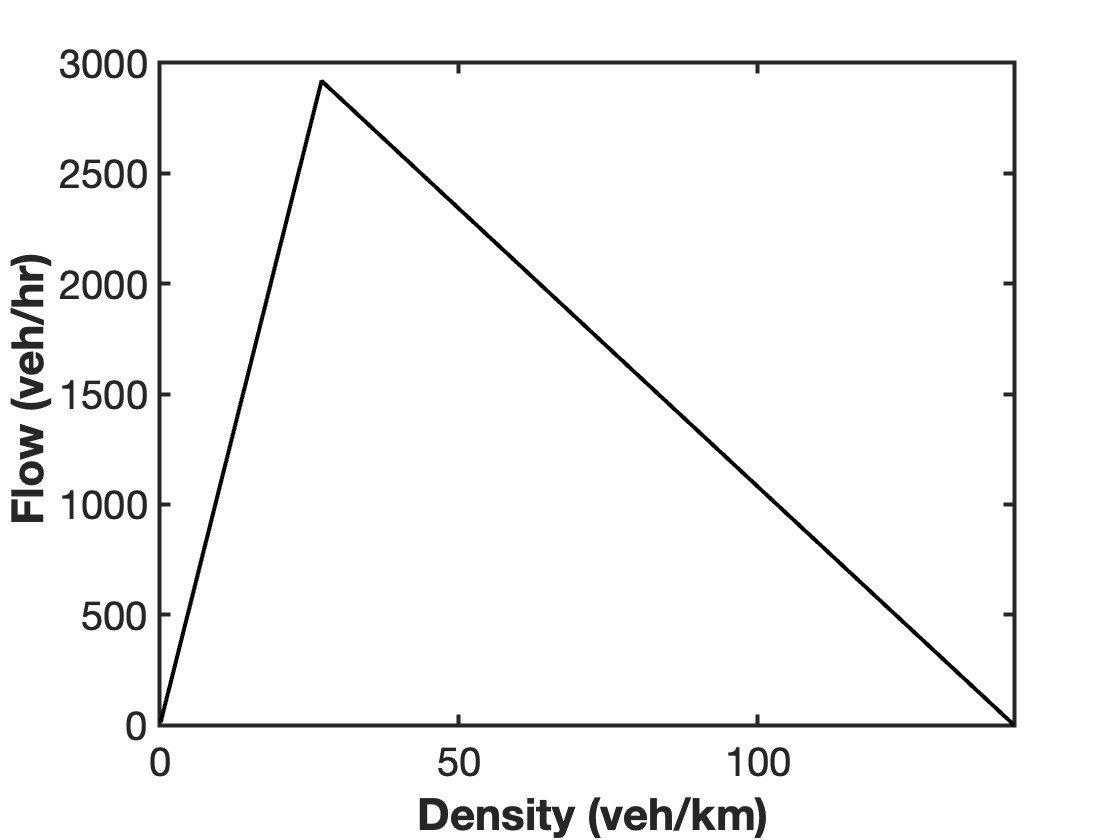}
}
    \caption{Plot of an optimal velocity function with $l=5$m and $v_{\max}=30$m/s and the corresponding triangular fundamental diagram.}
    \label{f2}
\end{figure}

During the simulation time, the leading CAV tries to maintain at an equilibrium speed but it encounters sinusoidal disturbance, which is similar to \cite{zhou2023autonomous}. The equilibrium speed is $v_0=15$m/s and the equilibrium headway is $h=22$m. And if we denote $p$ as the duration of each period of the sinusoidal disturbance and $A$ as the amplitude of the disturbance, then the velocity of the leading vehicle can be written as:
\begin{equation}
    v_N=v_0+A\sin\left(\frac{2\pi}{p}t\right),
\end{equation}
where the periods are set to $p=5, \;10, \;15, \;20$s, with an amplitude $A=5$m/s. OVM and P-OVM with sensitivity constant $a=1.2,\; 2.4$ are tested with a simulation duration is $60$ seconds for all cases. Simulation results are shown in Figures \ref{31}-\ref{34}. Figure \ref{31}, \ref{32} are the headway profiles of OVM and P-OVM with sensitivity constant $a=1.2$ under different perturbation frequency. Figure \ref{33}, \ref{34} are the headway profiles of OVM and P-OVM with sensitivity constant $a=2.4$ under different perturbation frequency. To further show the difference in the P-OVM cases, Table \ref{tab1} presents the average oscillations across all vehicles for various sensitivity constants and perturbations.
% plot of experiment with equilibrium speed

\begin{figure}
    \centering
    \subfloat[$p=20$]{
        \includegraphics[width=0.4\textwidth]{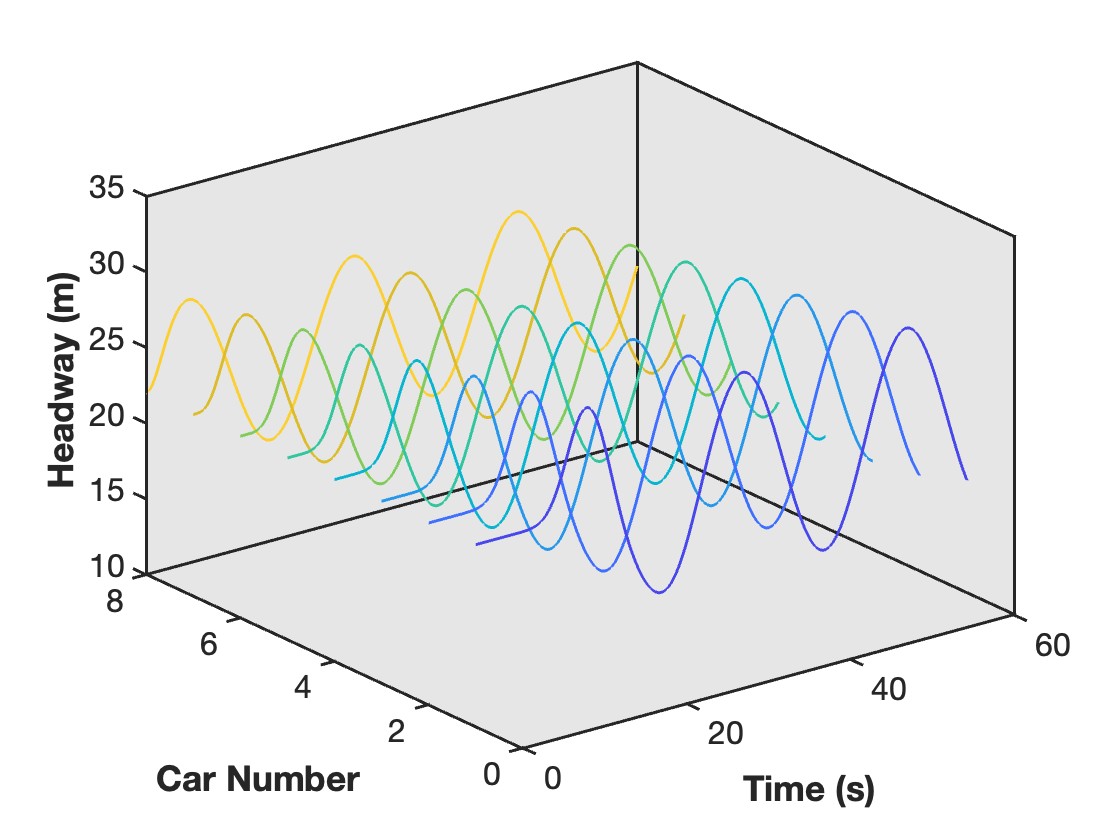}
    }
    \subfloat[$p=15$]{
        \includegraphics[width=0.4\textwidth]{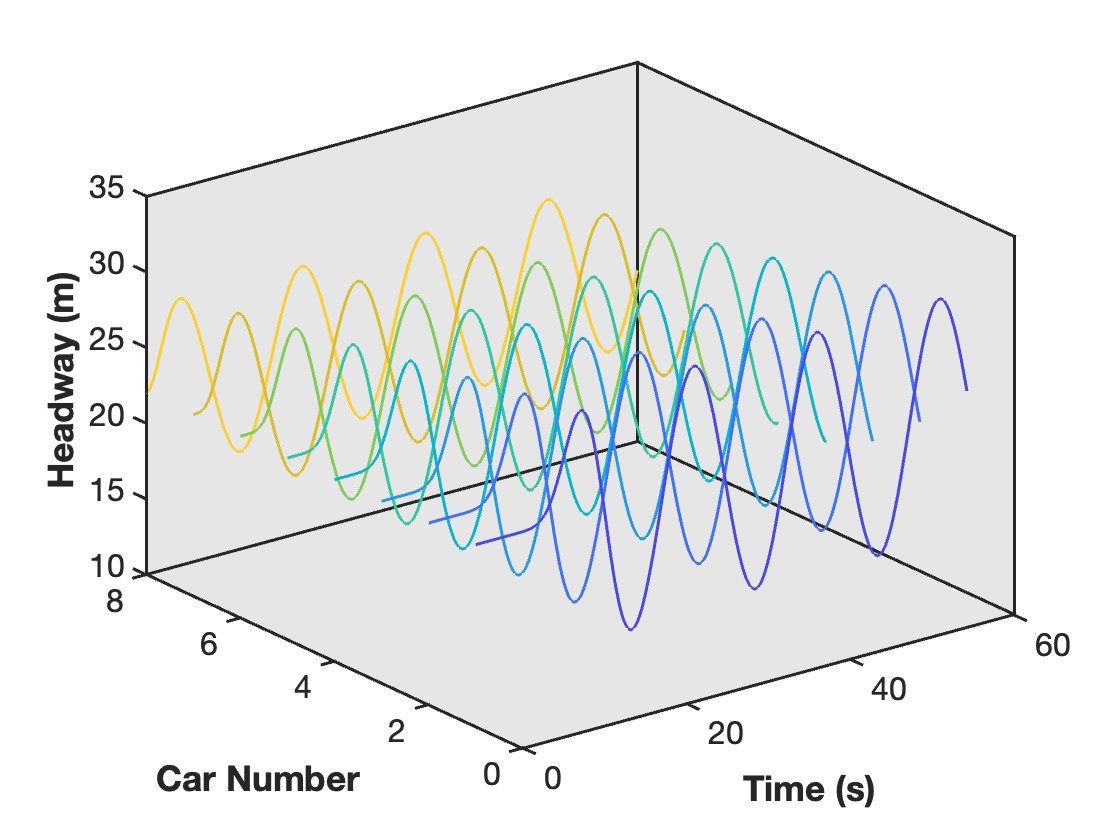}
    }
    \\
    \subfloat[$p=10$]{
        \includegraphics[width=0.4\textwidth]{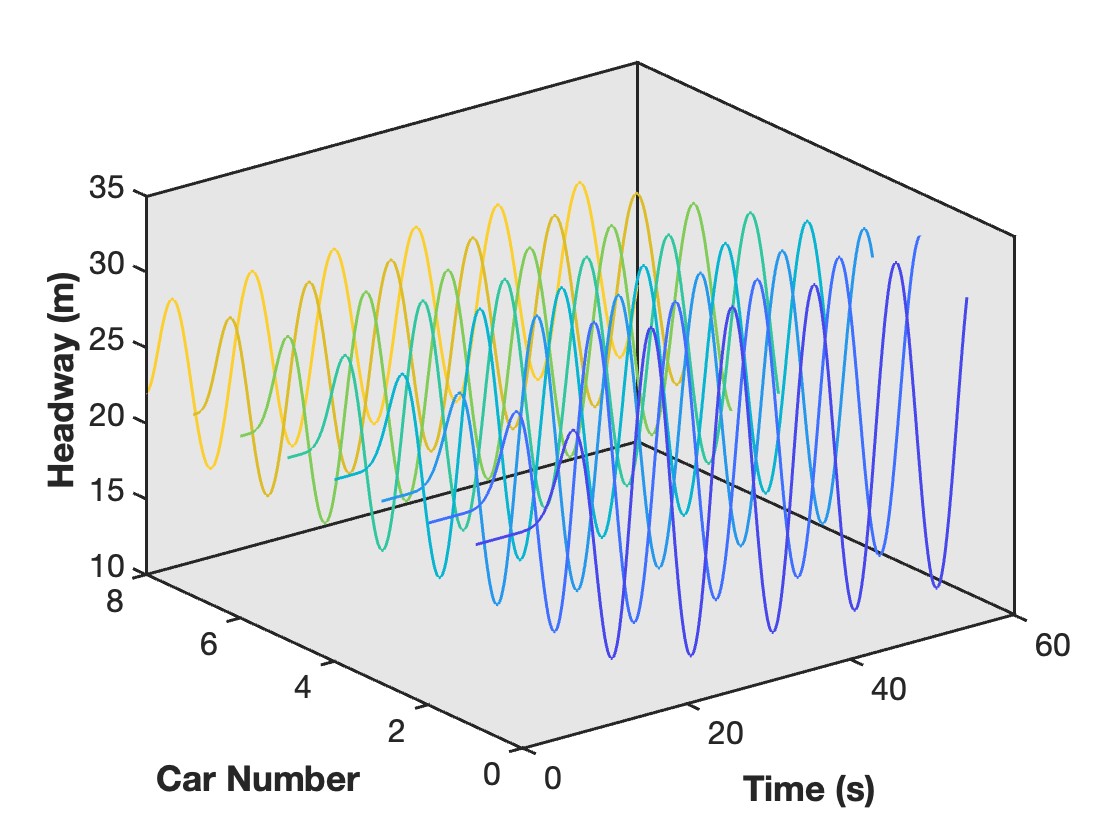}
    }
    \subfloat[$p=5$]{
        \includegraphics[width=0.4\textwidth]{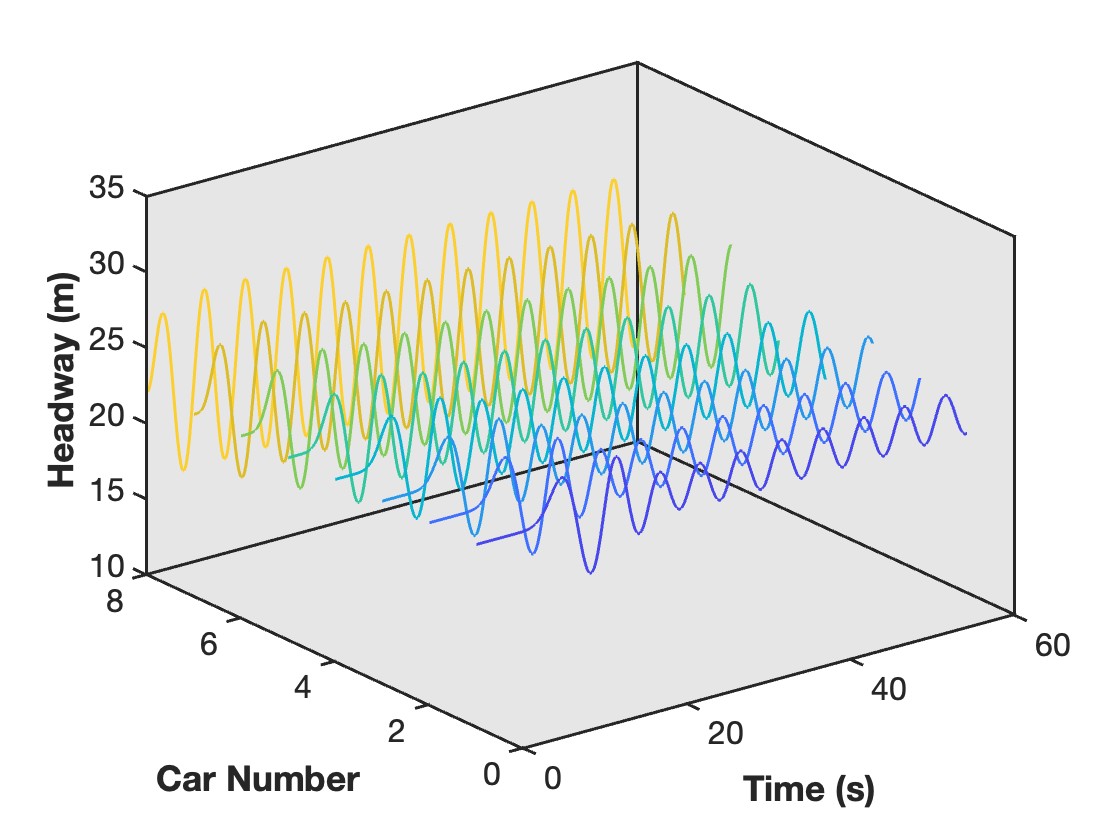}
    }
    \caption{Headway profile of OVM with sensitivity constant $a=1.2$ and period parameter $p=5, \;10, \;15, \;20$}
    \label{31}
\end{figure}

\begin{figure}
    \centering
    \subfloat[$p=20$]{
        \includegraphics[width=0.4\textwidth]{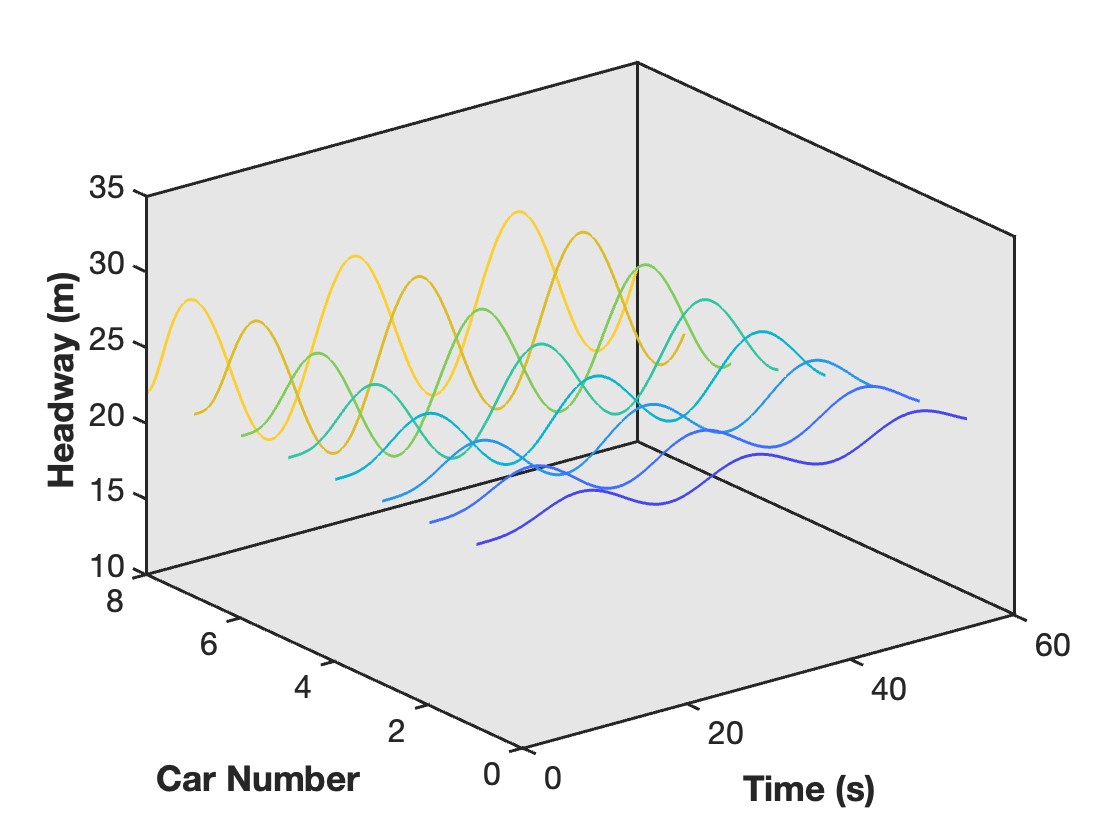}
    }
    \subfloat[$p=15$]{
        \includegraphics[width=0.4\textwidth]{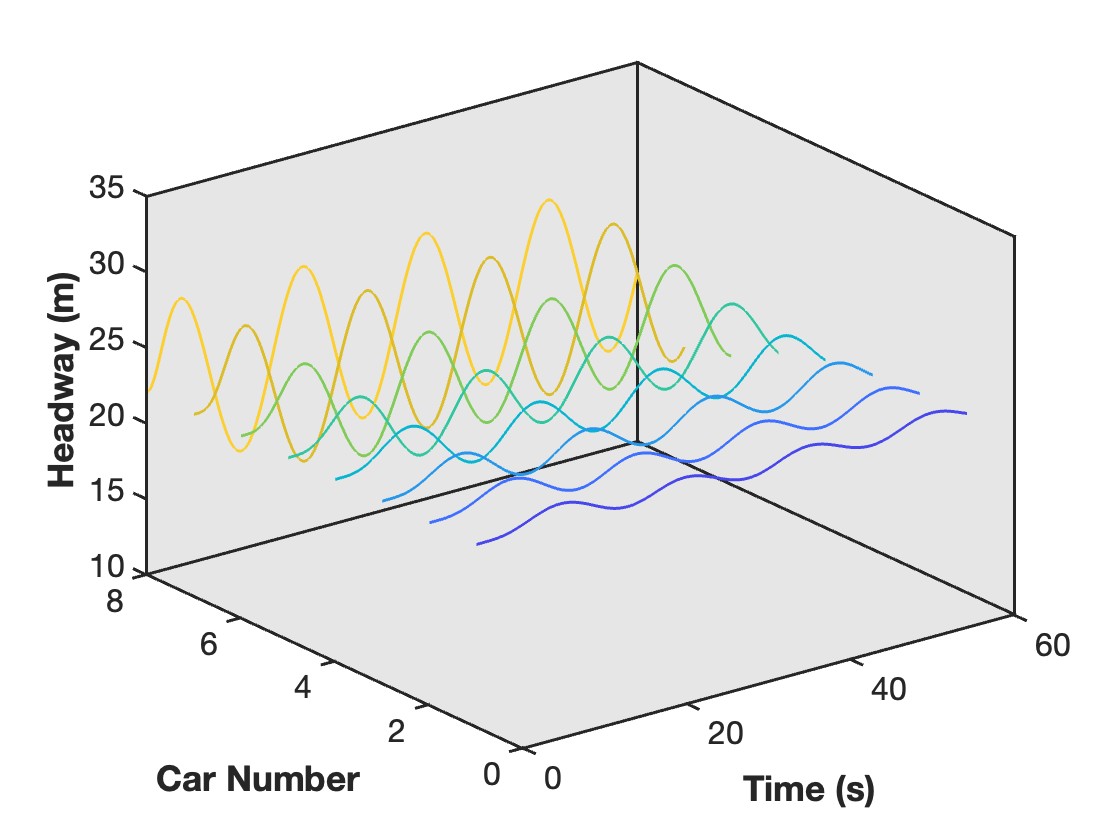}
    }
    \\
    \subfloat[$p=10$]{
        \includegraphics[width=0.4\textwidth]{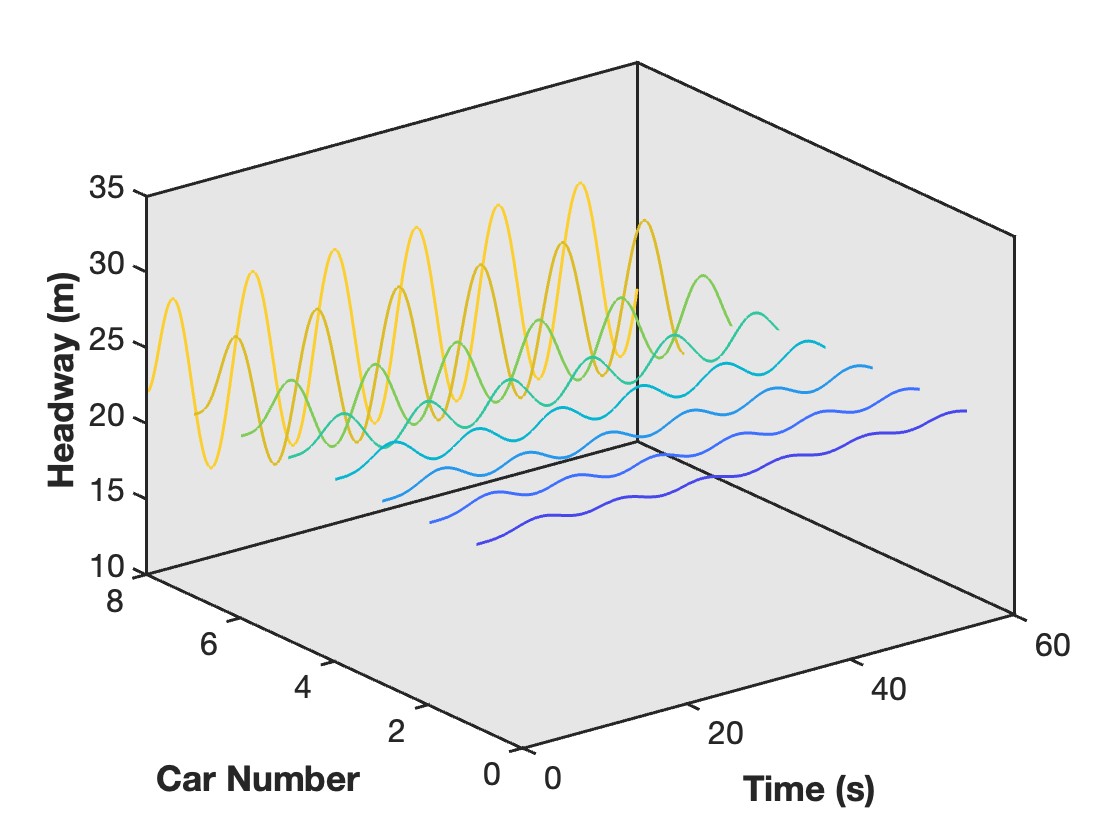}
    }
    \subfloat[$p=5$]{
        \includegraphics[width=0.4\textwidth]{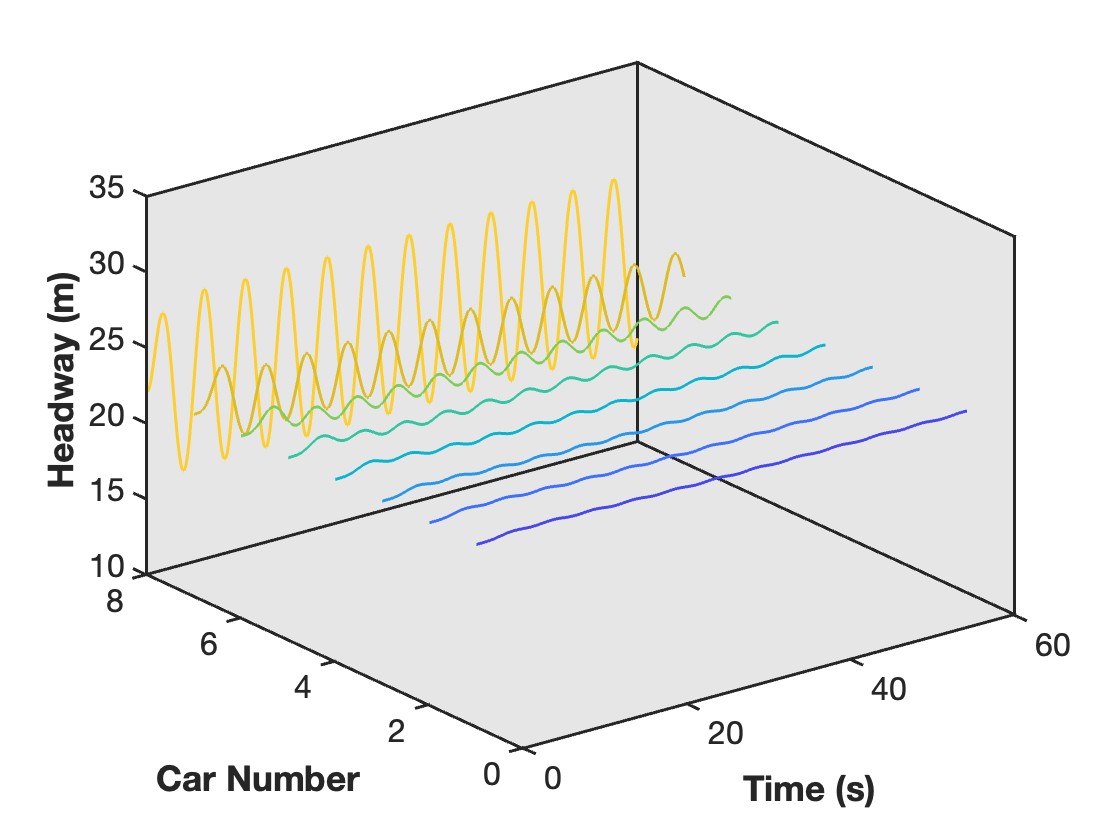}
    }
    \caption{Headway profile of P-OVM with sensitivity constant $a=1.2$ and period parameter $p=5, \;10, \;15, \;20$}
    \label{32}
\end{figure}

\begin{figure}
    \centering
    \subfloat[$p=20$]{
        \includegraphics[width=0.4\textwidth]{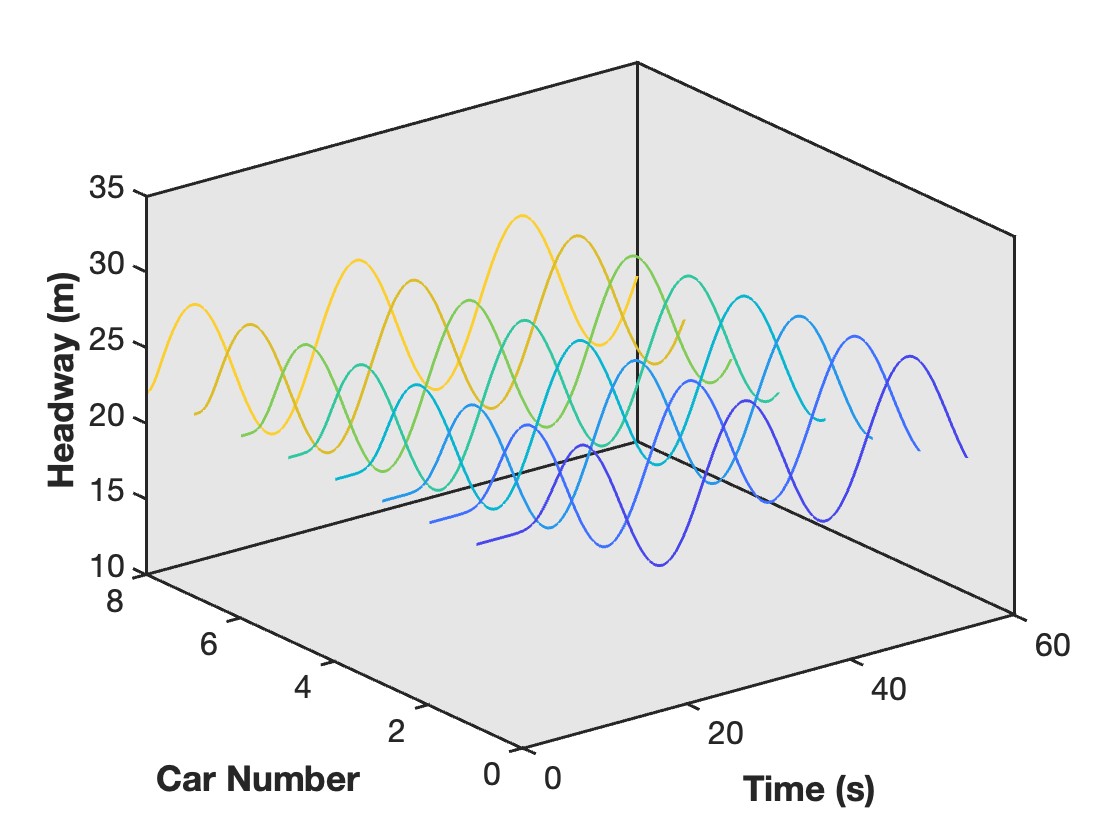}
    }
    \subfloat[$p=15$]{
        \includegraphics[width=0.4\textwidth]{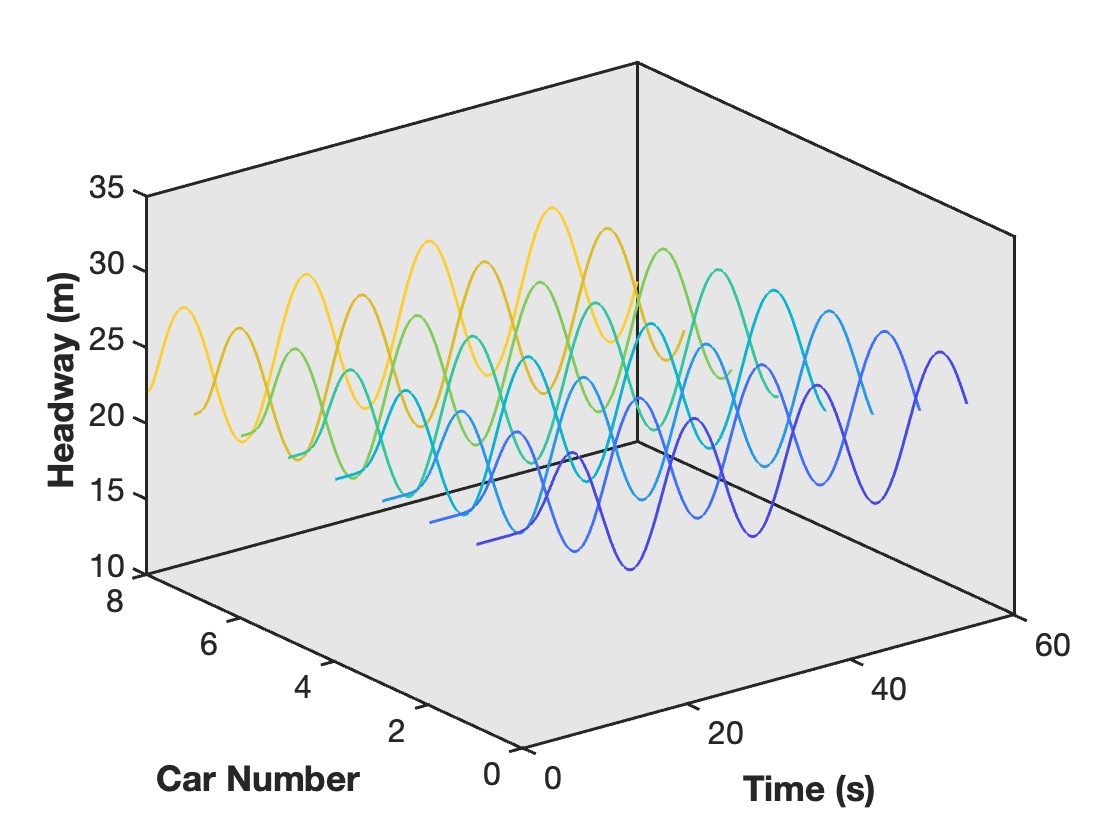}
    }
    \\
    \subfloat[$p=10$]{
        \includegraphics[width=0.4\textwidth]{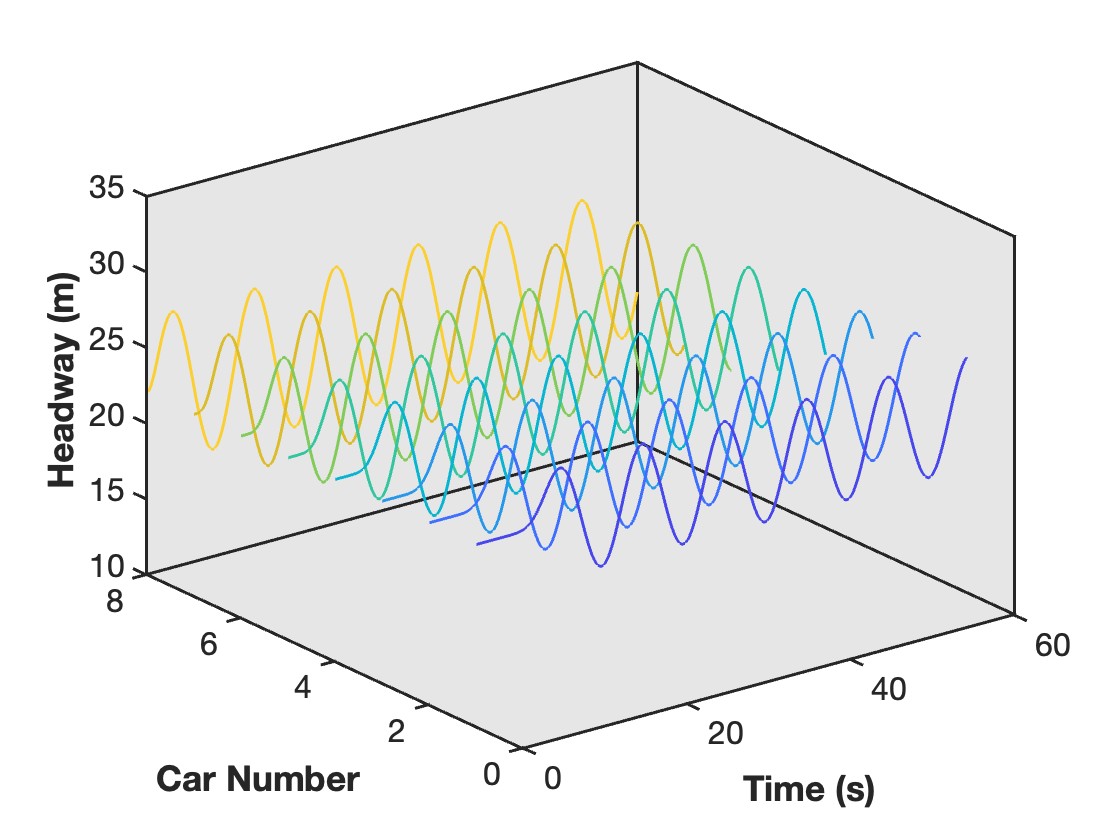}
    }
    \subfloat[$p=5$]{
        \includegraphics[width=0.4\textwidth]{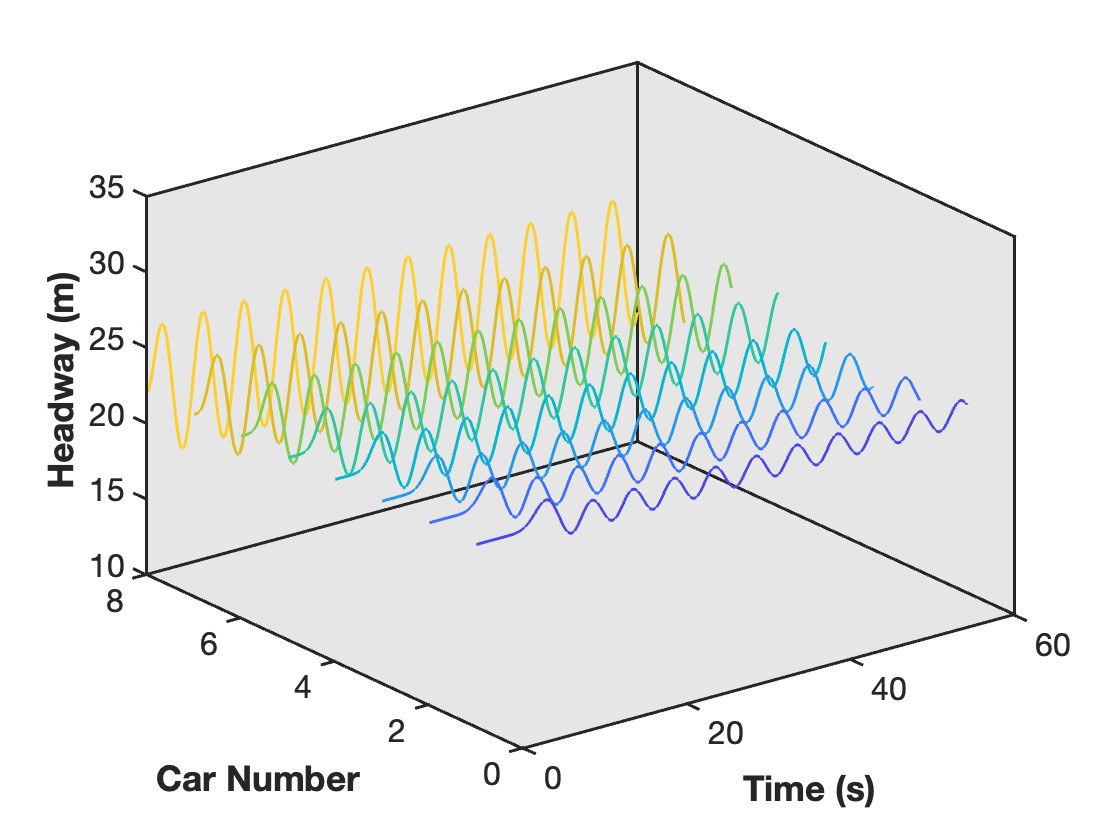}
    }
    \caption{Headway profile of OVM with sensitivity constant $a=2.4$ and period parameter $p=5, \;10, \;15, \;20$}
    \label{33}
\end{figure}

\begin{figure}
    \centering
    \subfloat[$p=20$]{
        \includegraphics[width=0.4\textwidth]{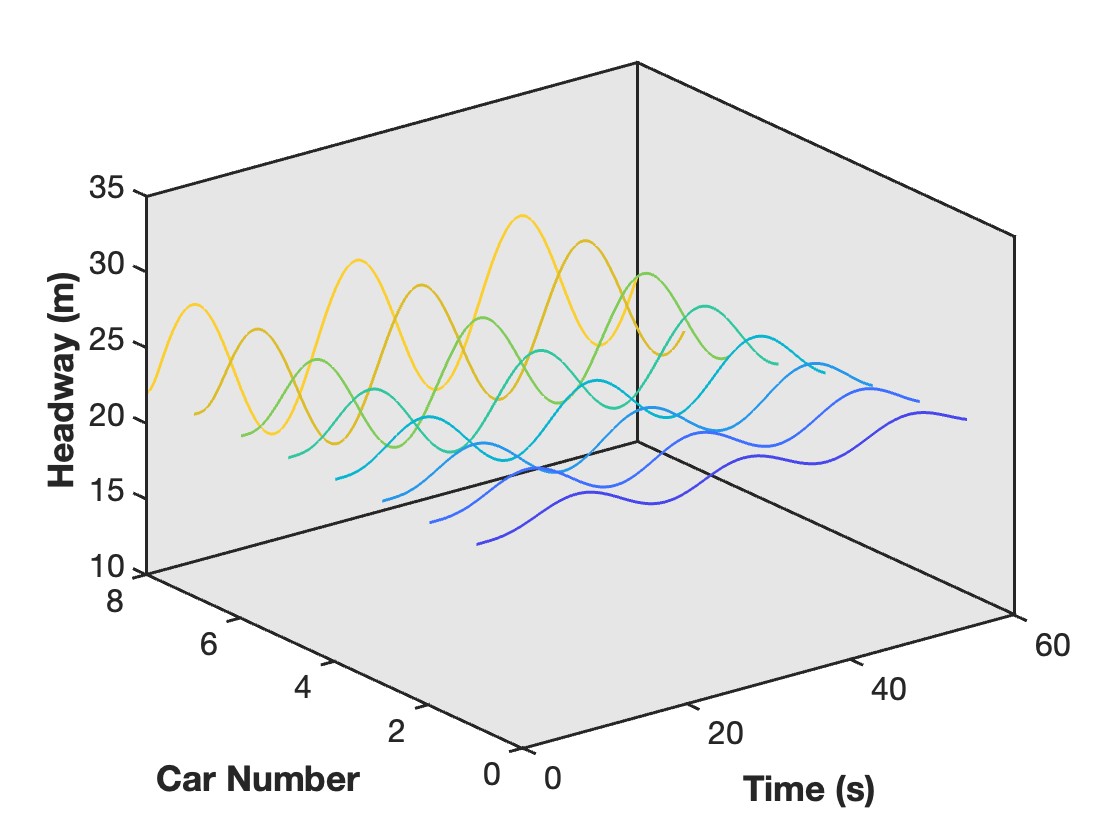}
    }
    \subfloat[$p=15$]{
        \includegraphics[width=0.4\textwidth]{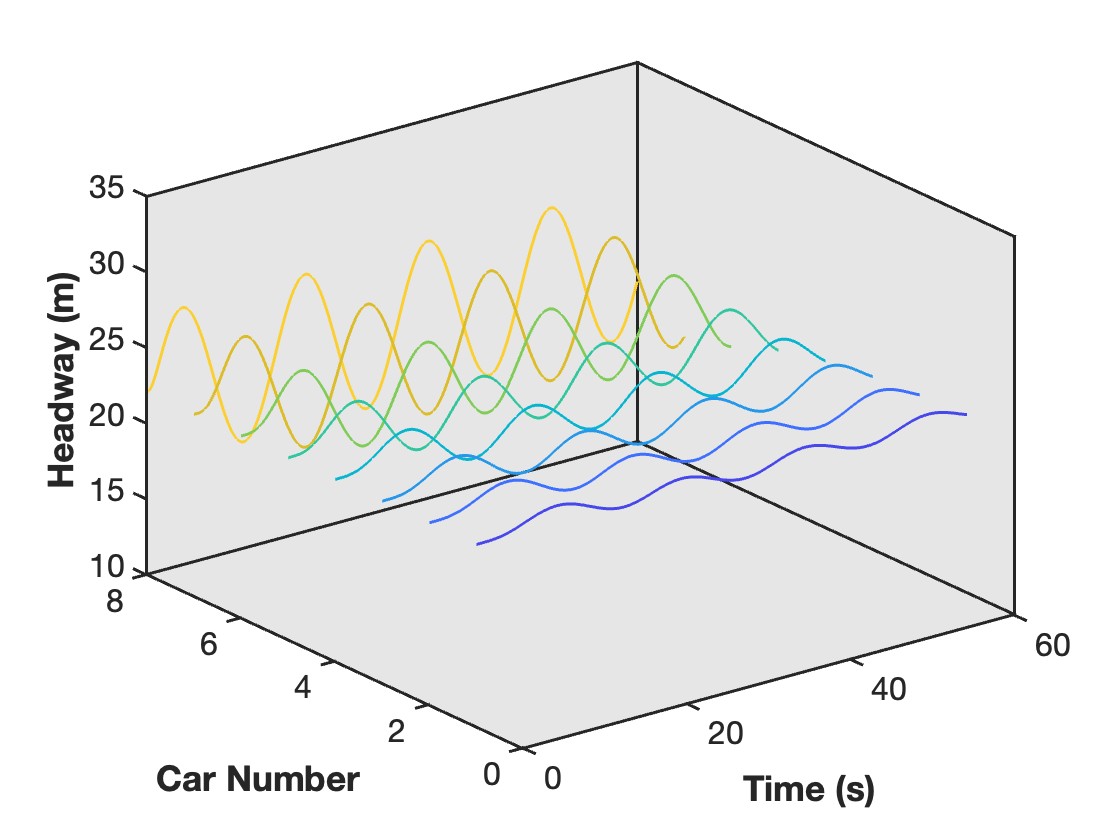}
    }
    \\
    \subfloat[$p=10$]{
        \includegraphics[width=0.4\textwidth]{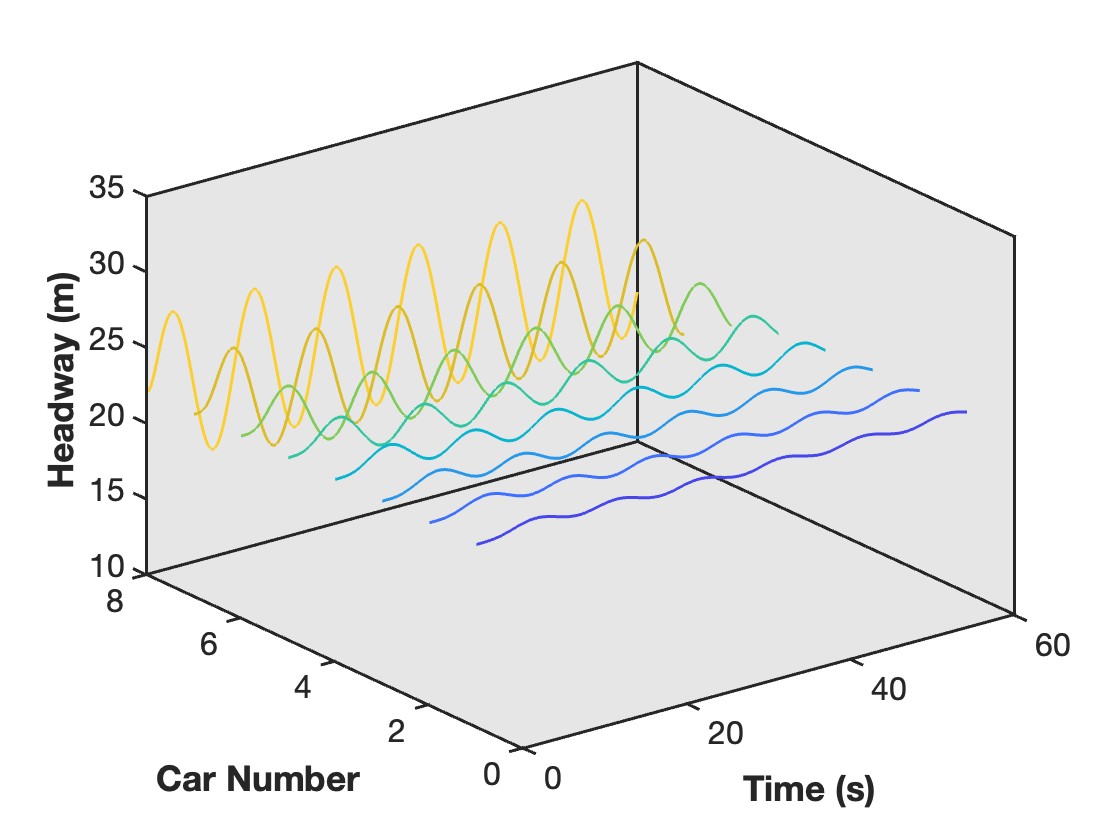}
    }
    \subfloat[$p=5$]{
        \includegraphics[width=0.4\textwidth]{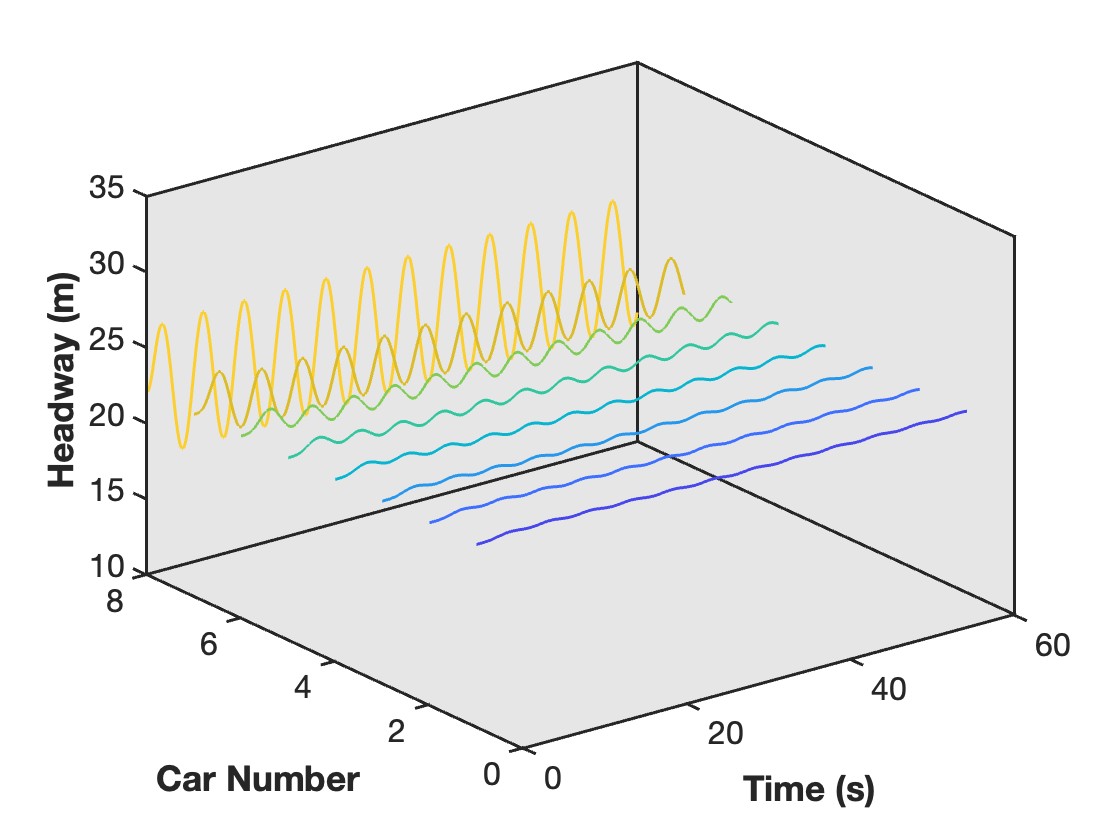}
    }
    \caption{Headway profile of P-OVM with sensitivity constant $a=2.4$ and period parameter $p=5, \;10, \;15, \;20$}
    \label{34}
\end{figure}

\begin{table}
    \centering
    \begin{tabular}{c|c|c|c|c}
        \toprule
        sensitivity/period & $p=5$ & $p=10$ & $p=15$ & $p=20$ \\
        \midrule
        $a=1.2$ & $0.5055$m & $0.8966$m & $1.1276$m & $1.3279$m \\
        \midrule
        $a=2.4$ & $0.4256$m & $0.7382$m & $0.9882$m & $1.2049$m \\
        \bottomrule
    \end{tabular}
    \caption{The average oscillations (in meters) across all vehicles for sensitivity constant $a=1.2,\;2.4$ and period parameter $p=5, \;10, \;20, \;30$}
    \label{tab1}
\end{table}

% explain simulation results
From the simulation results, we observe that when the frequency of the sinusoidal disturbance is high ($p=5$), the disturbance amplitude can decrease even without platoon control, despite the OVM being unstable at $a=1.2$. This is possibly due to the frequent changes of the optimal velocity and the delayed reaction inherent in OVM. However, with the P-OVM, the headways stabilize around the equilibrium state for all vehicles, and the oscillations are smaller. Stability improves slightly for $a=2.4$ compared to $a=1.2$, suggesting that a high level of sensitivity is not necessary for CAVs. In summary, P-OVM enhances string stability under periodic disturbances.
\FloatBarrier
\section{Conclusions \label{sec5}}
% New model, increased stability with mathematical proof, simulation discovery, limitation, future work
In this paper, we introduced two models: a platoon model with a coordinated following strategy (P-OVM) and a transition phase model with adjustable parameters (T-OVM). While P-OVM was theoretically proven to be always linearly stable, its practical application may be constrained by safety concerns when both control and state measurements are prone to errors. We also analyzed the stability of the base and transition phase models, providing conditions for their stability. To verify the theoretical findings, numerical simulations were performed on both a ring road with an initial disturbance and an infinite road with periodic disturbances. Several model parameters were tested and analyzed. The simulation results confirmed that the proposed platoon control guarantees linear stability for all sensitivity constants, with the constants influencing how quickly disturbances are attenuated within the platoon. For the transition phase model, with a fixed total sensitivity, stability improved with increased sensitivity to the leading vehicle, in line with theoretical predictions.

The proposed P-OVM and T-OVM significantly outperformed the HDVs' model (OVM) in suppressing disturbances and maintaining string stability. These models are simple forms of platooning that multi-following models tend to focus on, representing a significant improvement over previous multi-following models. With the rapid development of CAV technologies, the proposed platoon control strategy offers a promising solution for mitigating traffic oscillations.

This paper can be extended in various directions for future research, including: (1) incorporating communication delays into the proposed models and analyzing their impact on stability; (2) extend the proposed models to multiple CAV platoons with interactions; (3) examining more complex road conditions, such as multi-lane ring roads with mixed CAVs and HDVs, with adjustments to the control design for the leading vehicle; (4) conducting field experiments with CAVs to test the proposed control design under various traffic and road conditions.
%safety %delay %road scenario% mixed flow % visible simulation % generalise to ge
\section*{Declaration of competing interest}
The authors declare that there are no source of financial interests or personal relationships that could have any impacts of the work in this paper.

\bibliographystyle{unsrt}  
\bibliography{references} 

\begin{appendices}
\section*{\large Appendices}
\addcontentsline{toc}{section}{Appendices} % If you want it to appear in the table of contents
    \section{Proof of Theorem 3.1 \label{ap1}}
    In this section we give the proof of Theorem \ref{lin}:
    \begin{proof}
        Similar to proof of \ref{sta2}, we consider small deviation $y_i$ from the equilibrium solution and neglect higher order. Then we have
        \begin{equation}
            \ddot{y}_i(t)=a\left[V'(h)(y_{i+1}(t)-y_i(t))-\dot{y}_i(t)\right],\; y_{N+1}=y_1,
            \label{linear2}
        \end{equation}
        which is a linear ODE system with $N$ equations. Suppose that $\lambda$ is an eigenvalue of the system, and $\xi_i$ is the coefficient of $y_i$ with the term $e^{\lambda t}$, then simplified from \eqref{linear2}, $\lambda$ satisfies
    \begin{equation}
        \lambda^2+a\lambda-aV'(h)(\frac{\xi_{i+1}}{\xi_i}-1)=0
    \end{equation}
    This means $\xi_{i+1}/\xi_i=r$ is fixed for given $\lambda$, and note that $\prod_{i=1}^{N}(\xi_{i+1}/\xi_i)=1$, we have
    \begin{equation}
        r=e^{\frac{2\pi k}{N}i},\; k = 1, 2, \dots,N.
        \label{eigen}
    \end{equation}
    Therefore $\lambda$ satisfies
    \begin{equation}
        \lambda^2+a\lambda-aV'(h)(e^{\frac{2\pi k}{N}i}-1)=0,
    \end{equation}
    for some $k=1,2,\dots,N$. Then, to have a stable system of $y_n$, the real parts of any $\lambda$ should be negative. By considering the solution of $\lambda$ with smaller real parts we have
    \begin{equation}
        a-\sqrt{\frac{d+\sqrt{d^2+e^2}}{2}}>0,
        \label{condGen}
    \end{equation}
    where $d=a^2+4aV'(h)\cos\frac{2\pi k}{N}-4aV'(h)$ and $e=-4aV'(h)\sin\frac{2\pi k}{N}$. For the system to be stable, we need this condition to be satisfied for every $k$. Let $\theta\triangleq \frac{2\pi k}{N}$, then \eqref{condGen} can be simplified to
    \begin{equation}
        -V'(h)\cos^2\theta+a\cos\theta+V'(h)-a<0.
        \label{cond2}
    \end{equation}
    \eqref{cond2} holds if \eqref{ovmsta} holds, otherwise for N sufficiently large we can let $\cos\theta=a/2V'(h)$ and \eqref{cond2} will fail. 
    \end{proof}
    
    \section{Proof of Theorem 3.3 \label{ap2}}
    In this section we give the proof of Theorem \ref{transta}:
    \begin{proof}
        Again we consider small deviation $y_i$ from the equilibrium solution and neglect higher order. Then we have
        \begin{equation}
            \ddot{y}_i(t)=\begin{cases}
                a\left[V'(h)(y_{i+1}(t)-y_i(t))-\dot{y}_i(t)\right]+b\left[V'(h)\frac{y_{N}(t)-y_i(t)}{N-i}-\dot{y}_i(t)\right],\;\text{if } i\neq N; \\
                (a+b)\left[V'(h)(y_{i+1}(t)-y_i(t))-\dot{y}_i(t)\right],\;\text{if } i= N,
            \end{cases}
            \label{linear3}
        \end{equation}
        which is a linear ODE system with $N$ equations. Suppose that $\lambda$ is an eigenvalue of the system, and $\xi_i$ is the coefficient of $y_i$ with the term $e^{\lambda t}$, then simplified from \eqref{linear3}, $\lambda$ satisfies
        \begin{equation}
        \lambda^2+(a+b)\lambda-aV'(h)(\frac{\xi_{i+1}}{\xi_i}-1)-\frac{bV'(h)}{N-i}(\frac{\xi_{N}}{\xi_i}-1)=0, \;\text{if } i\neq N 
        \label{mix1}
        \end{equation}
        and
        \begin{equation}
            \lambda^2+(a+b)\lambda-(a+b)V'(h)(\frac{\xi_{1}}{\xi_N}-1)=0.
        \end{equation}
        When N is very large ($N\to\infty$), $bV'(h)/(N-i)$ is very small for most i, and the last term of \eqref{mix1} has minimal effect. This means $\xi_{i+1}/\xi_i$ will converge to eigenvalues given in \eqref{eigen}. Then $\lambda$ will converge to values given by
        \begin{equation}
            \lambda^2+(a+b)\lambda-aV'(h)(e^{\frac{2\pi k}{N}i}-1)=0.
        \end{equation}
        Follow the same steps of previous proof we can find the stability criterion \eqref{transtae}.
    \end{proof}
\end{appendices}
\end{document}